\documentclass[10pt,journal,compsoc]{IEEEtran}

 \usepackage{amssymb}
 \usepackage{amsmath,bbm}
 \usepackage{bm}
 \usepackage[linesnumbered,ruled,vlined]{algorithm2e}
 \usepackage{amsthm}
 \usepackage{enumerate}
      \usepackage{graphicx}
      \usepackage{epstopdf}
\usepackage{url}
\usepackage{color}
\usepackage{amssymb}

\usepackage{enumitem}
\usepackage{amsmath,textcomp,enumerate,bbm,latexsym}
\setlist[itemize,1]{leftmargin=1em,noitemsep,topsep=0pt}

\usepackage[colorlinks=true,
            citecolor=red,
            linkcolor=blue,
        anchorcolor=green,
            urlcolor=red]{hyperref}

\newtheorem{theorem}{Theorem}

\newtheorem{lemma}{Lemma}
\newtheorem{assumption}{Assumption}
\newtheorem{definition}{Definition}
\newtheorem{remark}{Remark}
\graphicspath{{images/}{../images/}}

\ifCLASSOPTIONcompsoc
  \usepackage[nocompress]{cite}
\else
  \usepackage{cite}
\fi

\ifCLASSINFOpdf
\else
\fi

\begin{document}
%

\title{Online Job Scheduling with Redundancy and Opportunistic Checkpointing: \\ A Speedup-Function-Based Analysis}

\author{Huanle Xu,~\IEEEmembership{Member,~IEEE,}
       Gustavo de Veciana,~\IEEEmembership{Fellow,~IEEE,}
       \\  Wing Cheong Lau,~\IEEEmembership{Senior Member,~IEEE,}
       Kunxiao Zhou

\IEEEcompsocitemizethanks{\IEEEcompsocthanksitem Huanle Xu and Kunxiao Zhou are with the School of Computer Science and Network Security, Dongguan University of Technology, Dongguan, Guangdong.
E-mail: \{xuhl,zhoukx\}@dgut.edu.cn.
\IEEEcompsocthanksitem Gustavo de Veciana is with the Department of Electrical and Computer Engineering, The University of Texas at Austin, Austin, TX, USA. E-mail: gustavo@ece.utexas.edu. 
\IEEEcompsocthanksitem Wing Cheong Lau is with the Department of Information Engineering, The Chinese University of Hong Kong, Shatin, N.T., Hong Kong. E-mail: wclau@ie.cuhk.edu.hk.

}
\thanks{Part of this work has been presented in IEEE Infocom 2017.}
}

\IEEEtitleabstractindextext{
\begin{abstract}
In a large-scale computing cluster, the job completions can be substantially delayed due to two sources of variability, namely, variability in the job size and that in the machine service capacity. To tackle this issue, existing works have proposed various scheduling algorithms which exploit redundancy wherein a job runs on multiple servers until the first completes. In this paper, we explore the impact of variability in the machine service capacity and adopt a rigorous analytical approach to design scheduling algorithms using redundancy and checkpointing. 
We design several online scheduling algorithms which can dynamically vary the number of redundant copies for jobs. We also provide new theoretical performance bounds for these algorithms in terms of the overall job flowtime by introducing the notion of a speedup function, based on which a novel potential function can be defined to enable the corresponding competitive ratio analysis. In particular, by adopting the online primal-dual fitting approach, we prove that our SRPT+R Algorithm in a non-multitasking cluster is $(1+\epsilon)$-speed, $\ O(\frac{1}{\epsilon})$-competitive.  We also show that our proposed Fair+R and LAPS+R($\beta$) Algorithms for a multitasking cluster are $(4+\epsilon)$-speed, $\ O(\frac{1}{\epsilon})$-competitive and {($2 + 2\beta + 2\epsilon)$-speed $O(\frac{1}{\beta \epsilon})$-competitive} respectively.  We demonstrate via extensive simulations that our proposed algorithms can significantly reduce job flowtime under both the non-multitasking and multitasking modes. 
\end{abstract}
\begin{IEEEkeywords}
Online Scheduling, Redundancy, Optimization, Competitive Analysis, Dual-Fitting, Potential Function
\end{IEEEkeywords}
}


\maketitle

\IEEEdisplaynontitleabstractindextext
\IEEEpeerreviewmaketitle

\IEEEraisesectionheading{\section{Introduction}}

\IEEEPARstart{J}{ob} traces from large-scale computing clusters indicate that the completion time of jobs can vary substantially \cite{Outliers,grass}. This variability has two sources: variability in the job processing requirements and variability in machine service capacity. The job profiles in production clusters also become increasingly diverse as small latency-sensitive jobs coexist with large batch processing applications which take hours to months to complete \cite{spark-framework}.  With the size of today's computing clusters continuing to grow, component failures and resource contention have become a common phenomenon in cloud infrastructure \cite{trace-archive,failure-modeling}. As a result, the rate of machine service capacity may fluctuate significantly over the lifetime of a job. The same job may experience a far higher response time when executed at a different time on the same server \cite{redundacy_d}. These two dimensions of variability make efficient job scheduling for fast response time (also referred to as job flowtime) over large-scale computing clusters challenging. 

To tackle variability in job processing requirements, various schedulers have been proposed to provide efficient resource sharing among heterogeneous applications.  Widely deployed schedulers to-date include the Fair scheduler \cite{fair-scheduler} and the Capacity scheduler \cite{capacity-scheduler}. It is well known that the Shortest Remaining Processing Time scheduler (SRPT) is optimal for minimizing the overall/ average job flowtime \cite{SRPT} on a single machine in the clairvoyant setting, i.e., when job processing times are known a priori. As such, many works have aimed to extend SRPT scheduling to yield efficient scheduling algorithms in the multiprocessor setting with the objective of reducing job flowtimes for different systems and programming frameworks \cite{Flow_Shops,Schedulers,overlapping_phases,resource-packing}. Under SRPT, job's residual precessing times are known to the job scheduler upon arrival and smaller jobs are given priority. However, if only the distribution of job sizes is known, it is shown in \cite{gittins-index} that, Gittins index-based policy is optimal for minimizing the expected job flowtime under the Poission job arrivals in the single-server case. The Gittins index depends on knowing the service already allocated to each job and gives priority to the job with the highest index. 



To deal with component failures and resource contention, computing clusters are exploiting redundant execution wherein multiple copies of the same job execute on available machines until the first completes. With redundancy, it is expected that one copy of the same job might complete quickly to avoid long completion times. For the Google MapReduce system, it has been shown that redundancy can decrease the average job flowtime by 44\% \cite{mapreduce:google}. Many other cloud computing systems apply simple heuristics to use redundancy and they have proven to be effective at reducing job flowtimes via practical deployments, e.g., \cite{mapreduce:google,Cloning,hadoop,Dryad,Performance,Outliers,Smart_Speculative}.



Recently, researchers have started to investigate the effectiveness of scheduling redundant copies from a queuing perspective \cite{Ashish_redundant,chen_queuing_coding,zhan_replication,zhan_infocom,shah_redundant,redundacy_d}. These works assume a specific distribution of the job execution time where jobs follow the same distribution. However, they do not characterize the major cause leading to the variance of the job response time, namely, whether the variance is due to variability of job size or to variability in machine service capacity. In fact, if there is no variability in the machine service capacity, making multiple copies of the same job may not help and redundancy is a waste of resource. 

To overcome the aforementioned limitations, we have developed a stochastic framework in our previous work \cite{machine_variability_MapReduce} to explore the impact of variability in the machine service capacity. In this framework, the service capacity of each machine over time is modeled as a stationary process. To take full advantage of redundancy, \cite{machine_variability_MapReduce} allows checkpointing \cite{Pruyne_managingcheckpoints} to preempt, migrate and perform dynamic partitioning \cite{Squillante95onthe} on its running jobs. By checkpointing, we mean the runtime system of a cluster takes a snapshot of the state of a job in progress so that its execution can be resumed from that point in the case of subsequent machine failure or job preemption \cite{application-checkpointing}. Upon checkpointing, the state of the redundant copy which has made the most progress is propagated and cloned to other copies. In other words, all the redundant copies of a job can be brought to that most advance state and proceed to execute from this updated state.

A fundamental limitation of \cite{machine_variability_MapReduce} is that checkpointing needs to be done periodically when a job is being processed. Such a checkpointing mechanism would incur large overheads when the cluster size is large while the scheduler needs to make scheduling decisions quickly. To tackle this limitation, in this paper, we limit the total number of checkpointings for each job. Moreover, we only allow checkpointing to occur on a job only if there is an arrival to or departure from the system. As such, the resultant algorithms are more scalable and applicable to real world implementations. 



Most previous works studying job scheduling assume that clusters are working in the non-multitasking mode, i.e., each server (CPU Core) in the cluster can only serve one job at any time. However, multitasking is a reasonable model of current scheduling policies in CPUs, web servers, routers, etc \cite{resource-sharing,resource_sharing_cpu,cpu-sharing}. In a multitasking cluster, each server may run multiple jobs simultaneously and jobs can share resources with different proportions. In this paper, we will also study scheduling algorithms, which determine checkpointing times, the number of redundant copies between successive checkpoints as well as the fraction of resource to be shares in both of the multitasking and non-multitasking settings.

\subsection*{Our Results} 
For non-multitasking clusters, we propose the SRPT+R algorithm where redundancy is used only when the number of active jobs is less than the number of servers. For clusters allowing multitasking, we first design the Fair+R Algorithm, which shares resources near equally among existing jobs, with priority given to jobs which arrived most recently. We then extend Fair+R Algorithm to yield the LAPS+R($\beta$) Algorithm, which only shares resources amongst a fixed fraction of the active jobs.
In summary, this paper makes the following technical contributions: 
\begin{itemize}

\item \textit{New Framework.} We present the first optimization framework to address the job scheduling problem with redundancy, subject to limited number of checkpointings. Our optimization problems consider both the multitasking and non-multitasking scenarios. 

\item \textit{New Techniques.} We introduce the notion of speedup functions in both the multitasking and non-multitasking cases. Thanks to this, we develop a new dual-fitting approach to bound the competitive performance for both SRPT+R and Fair+R. Based on the speedup function, we also design a novel potential function accounting for redundancy to analyze the performance of LAPS+R($\beta$) in the multi-tasking setting. By changing the speedup function, one can readily apply our dual-fitting approach as well as the potential function analysis to other resource allocation problems in the multi-machine setting with/ without multitasking. 

\item \textit{New Results.} Under our optimization framework, SRPT+R achieves a much tighter competitive bound than other SRPT-based redundancy algorithms under different settings, e.g., \cite{machine_variability_MapReduce}. Moreover, LAPS+R($\beta$) is the first one to address the redundancy issue among those algorithms which work under the multitasking mode. 

\end{itemize}

The rest of this paper is organized as follows. After reviewing the related work in Section \ref{related_work}, we introduce our system model and optimization framework in Section \ref{system_model_section}. In Section \ref{no-resource-sharing}, we present SRPT+R and its performance bound in a non-multitasking cluster. We proceed to introduce the design and analysis for both Fair+R and LAPS+R($\beta$) under the multitasking mode in Section \ref{resource-sharing}.  Before concluding our work in Section \ref{conclusions}, we conduct several numerical studies in Section \ref{numerical-study} to evaluate our proposed algorithms.

\vspace{-.5em}
\section{Related Work}
\label{related_work}
In this{} section, we begin by giving a brief introduction to existing work on job schedulers. Then, we review the related work on redundancy schemes in large-scale computing clusters presented by priori research from the industry and academia. 

The design of job schedulers for large-scale computing clusters is currently an active research area \cite{Fast_completion,Flow_Shops,Joint_Phase,Joing_scheduling,overlapping_phases,Schedulers}. In particular, several works have derived performance bounds towards minimizing the total job completion time \cite{Fast_completion,Joing_scheduling,Joint_Phase} by formulating an approximate linear programming problem. By contrast, \cite{srpt-approximation} shows  that there is a strong lower bound on any online randomized algorithm for the job scheduling problem on multiple unit-speed processors with the objective of minimizing the overall job flowtime. Based on this lower bound, some works extend the SRPT scheduler to design algorithms that minimize the overall flowtimes of jobs which may consist of multiple small tasks with precedence constraints \cite{Flow_Shops,Schedulers,Joing_scheduling,overlapping_phases}. The above work was conducted in the clairvoyant setting, i.e., the job size is known once the job arrives. For the non-clairvoyant setting, \cite{dual-fitting-polyhedral,dual-fitting-migrate,round-robin} design several multitasking algorithms under which machines are allocated to all jobs in the system and priorities are given to jobs which arrive most recently. All of the above studies assume accurate knowledge of machine service capacity and hence do not address dynamic scheduling of redundant copies for a job.

Production clusters and big data computing frameworks have adopted various approaches to use redundancy for running jobs. The initial Google MapReduce system launches redundant copies when a job is close to its completion \cite{mapreduce:google}. Hadoop adopts another solution called LATE, which schedules a redundant copy for a running task only if its estimated progress rate is below certain threshold \cite{hadoop}. By comparison, Microsoft Mantri {\cite{Outliers}} schedules a new copy for a running task if its progress is slow and the total resource consumption is expected to decrease once a new redundant copy is made. 

Researchers have proposed different schemes to take advantage of redundancy via more careful designs. For example,  {\cite{Smart_Speculative}} proposes a smart redundancy scheme to accurately estimate the task progress rate and launch redundant copies accordingly. The authors in {\cite{Cloning}} propose to use redundancy for very small jobs when the extra loading is not high.  As an extension to {\cite{Cloning}}, they further develops GRASS \cite{grass}, which carefully schedules redundant copies for approximation jobs. Moreover, \cite{hopper_sigcomm} proposes Hopper to allocate computing slots based on the virtual job size, which is larger than the actual size. Hopper can immediately schedule a redundant copy once the progress rate of a task is detected to be slow. No performance characterization has been developed for these heuristics.

In our previous work, we have developed several optimization frameworks to study the design of scheduling algorithms utilizing redundancy \cite{huanle_task_cloning,speculative-multi-job}. The proposed algorithms in \cite{speculative-multi-job} require the knowledge of exact distribution of the task response time. We also analyze performance bounds of the proposed algorithm which extends the SRPT Scheduler in \cite{huanle_task_cloning} by adopting the potential function analysis. A fundamental limitation is that these resultant bounds are not scalable as they increase linearly with the number of machines. Recently,  \cite{gardner2016better}  proposes a simple model to address both machine service variability and job size variability. However, \cite{gardner2016better} only considers the FIFO scheduling policy on each server to characterize the average job response time from a queuing perspective.


Another body of research related to this paper focuses on the study of scheduling algorithms for jobs with intermediate parallelizability. In these works, e.g., \cite{SRPT-redundancy-sharing,scalably-scheduling,scaling-multiprocessor,Gupta10schedulingjobs,SRPT-transient}, jobs are parallelizable and the service rate can be arbitrarily scaled. In particular, Samuli \textit{et al.} present several optimal scheduling policies for different capacity regions in \cite{SRPT-transient} but for the transient case only. \cite{scalably-scheduling} \cite{scaling-multiprocessor} and \cite{Gupta10schedulingjobs} propose similar multitasking algorithms for jobs wherein priorities are given to jobs which arrive the most recently. These works develop competitive performance bounds with respect to the total job flowtime by adopting potential function arguments. \cite{SRPT-redundancy-sharing} also provides a competitive bound for the SRPT-based parallelizable algorithm in the  multitasking setting.  One limitation of \cite{SRPT-redundancy-sharing} is that the resultant bound is potentially very large. 
By contrast, this paper is motivated by the setting where there is variability in the machine service capacity.


For the analysis of SRPT+R algorithm in Section \ref{proof_dual_fitting_technique}, and Fair+R algorithm in \ref{Fair_R_algorithm}, we adopt the dual fitting approach. Dual fitting was first developed by \cite{dual-fitting,dual-fitting-non-linear} and is now widely used for the analysis of online algorithms \cite{dual-fitting-polyhedral,dual-fitting-migrate}. In particular, \cite{dual-fitting} and \cite{dual-fitting-polyhedral,dual-fitting-migrate} address linear objectives, and use the dual-fitting approach to derive competitive bounds for traditional scheduling algorith{}ms wi{}thout redundancy. By contrast, \cite{dual-fitting-non-linear} focuses on a convex objective in the multitasking setting. By comparison, this paper includes integer constraints associated with the non-multitasking mode. Moreover, our setting of dual variables is novel in the sense that it deals with the dynamical change of job flowtime across multiple machines where other settings of dual variables can only deal with the change of job flowtime on one single machine.

We apply the potential function analysis to bound the performance of LAPS+R($\beta$) in Section \ref{resource-sharing}. Potential function is widely used to derive performance bounds with resource augmentation for online parallel scheduling algorithms e.g., \cite{SRPT-redundancy-sharing,scalably-scheduling}. However, since we need to deal with redundancy and checkpointing, the design of our potential function is totally different from that in \cite{scalably-scheduling} and \cite{SRPT-redundancy-sharing} which only address sublinear speedup. 

While this paper adopts a framework similar to the one in \cite{machine_variability_MapReduce} to model machine service variability, it differs from \cite{machine_variability_MapReduce} in two major aspects. Firstly, the requirement of limiting the the total number of checkpointings results in a very different optimization problem which is much more difficult to solve from the one in \cite{machine_variability_MapReduce}. To tackle this challenge, in this paper, we adopt both the dual fitting approach and potential function analysis to make approximations and bound the competitive performance. By contrast, \cite{machine_variability_MapReduce} only applies the potential function analysis to derive performance bounds. Secondly, the current paper considers both the multi-tasking mode and non-multitasking mode to design corresponding online scheduling algorithms using redundancy. By contrast, the scheduling algorithms proposed in \cite{machine_variability_MapReduce} can only work under the non-multitasking mode.

\section{System Model}
\label{system_model_section}
Consider a computing cluster which consists of $M$ servers (machines) where the servers are indexed from $1$ to $M$. Job $j$ arrives to the cluster at time $a_j$ and the job arrival process, $(a_1,a_2,\cdots,a_N)$, is an arbitrary deterministic time sequence. In addition, job $j$ has a workload which requires $p_j$ units of time to complete when processed on a machine at {\em{unit}} speed. 
Job $j$ completes at time $c_j$ and its flowtime $f_j$, is denoted by $f_j = c_j - a_j$. In this paper, we focus on minimizing the overall job flowtime, i.e., $\sum_{j=1}^{N}f_j$. 

The service capacity of machines are assumed to be identically distributed random processes with stationary increments. To be specific, we let $S_i = (S_i(t)|t \geq 0)$ be a random process where $S_i(t,\tau) = S_i(\tau) - S_i(t)$ denote the {\em{cumulative}} service delivered by machine $i$ in the interval $(t,\tau]$. The service capacity of a machine has unit mean speed and a peak rate of $\Delta$, so for all $\tau > t \geq 0$, we have $S_i(t,\tau] \leq  (\tau - t) \cdot \Delta$ almost surely and $\mathbbm{E}\big[S_i(t,\tau]\big]=\tau - t$.



In this paper, our aim is to mitigate the impact of service variability by (possibly) varying the number of redundant copies with appropriate checkpointing. Checkpointing can make the most out of the allocated resources, i.e., start the processing of the possibly redundant copies at the most advanced state amongst the previously executing copies. In fact, we shall make the following assumption across the system:
\begin{assumption}
\label{assumption_1}
A job $j$ can be checkpointed only if there is an arrival to, or departure from, the system.

\end{assumption}

\begin{remark}
We refer to Assumption \ref{assumption_1} as a scalability assumption as it limits the checkpointing overheads in the system. 
\end{remark}


Below, we will first introduce a service model where each server can only serve one job at a time. In Section \ref{resource_sharing}, we will discuss a service model which supports multitasking, i.e., a server can execute multiple jobs simultaneously.

\subsection{Job processing in a Non-Multitasking Cluster}
\label{non_resource_sharing}
As illustrated in Fig.~\ref{system_model}, one can view the service process of job $j$ in a non-multitasking cluster by dividing its service period (from its arrival to its completion) into several subintervals, i.e., $\big\{(t^{k-1}_j,t^{k}_j]\big\}_k$ where $t^{k}_j$ denotes the time when the $k$th checkpointing of job $j$ occurs. The job arrival and completion times are also considered as checkpointing times, i.e., $t^{0}_j = a_j$ and $t^{L_j}_j = c_j$ if job $j$ experiences $(L_j + 1)$ checkpoints. During in $(t^{k-1}_{j},t^{k}_j]$, job $j$ is running on $r^k_j$ redundant servers. Thus, together $\bm{t_j} = \big(t^k_j \big|k=0,1,\cdots,L_j \big)$ and $\bm{r_j} = \big(r^k_j \big|k=1,2,\cdots,L_j \big)$ capture the checkpoint times and the scheduled redundancy for job $j$. 


\begin{figure}
\centering
\includegraphics[width=0.48\textwidth]{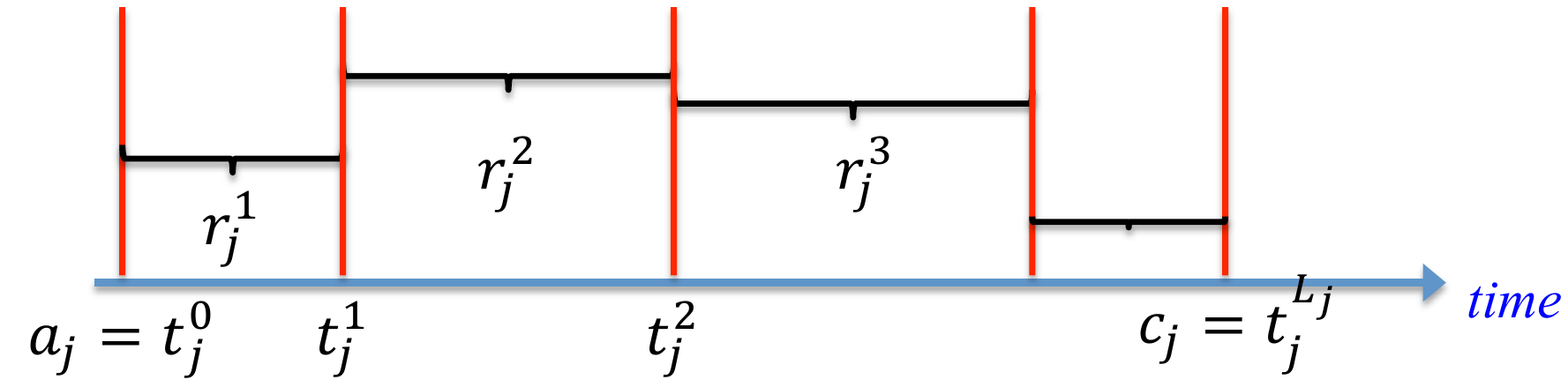}
\vspace{-1em}
\caption{The service process of job $j$.}
\label{system_model}
\vspace{-1.5em}
\end{figure}

We will let $g(r,t)$ denote the cumulative service delivered to a job on $r$ redundant machines and checkpointed at the end of an interval of duration $t$.  Clearly, $g(r,t)$ is equivalent to the amount of work processed by the redundant copy which has made the most progress. In this paper, we make the following assumption for $g(r,t)$: 
\begin{assumption}
\label{assumption_approximation}
We shall model (approximate) the cumulative service capacity under redundant execution, $g(r,t)$, by its mean, i.e.,
\begin{equation}
g(r,t) = \mathbbm{E}\Big[\max_{i=1,2,\cdots,r}  S_i(0,t]  \Big].
\end{equation}
\end{assumption}

\begin{remark}
Assumption \ref{assumption_approximation} essentially replaces the service capacity of the system with the mean but accounts for the mean gains one might expect when there are redundant copies executed.  
\end{remark}

The following lemmas illustrate two important properties of $g(r,t)$:

\begin{lemma}
\label{lemma_redundancy_concave}
For a fixed $t$, $\{g(r,t)\}_r$ is a concave sequence, i.e., $g(r,t) - g(r-1,t) \leq g(r-1,t) - g(r-2,t)$. 
\end{lemma}
\begin{proof}
Let $H_r(0,t] = \max_{i=1,2,\cdots,r}S_i(0,t]$ and define $F(x,t)$ as the cumulative distribution function of random variable $S_i(0,t]$ for a fixed $t$. Thus, we have $\Pr(H_r(0,t] \leq x) = F^r(x,t)$ and $g(r,t) =  \mathbbm{E}\big[H_r(t)\big] = \int_{0}^{\infty} (1-F^r(x,t)) dx$,  which further implies that: 
\begin{equation}
\begin{split}
g(r,t) - g(r-1,t) & = \int_{0}^{\infty} F^{r-1}(x,t) \cdot (1-F(x,t)) dx \\  & \leq \int_{0}^{\infty} F^{r-2}(x,t) \cdot (1-F(x,t)) dx \\ & = g(r-1,t) - g(r-2,t).
 \end{split}
  \label{equation_concaviety_g}
\end{equation}
This completes the proof. 
\end{proof}

Lemma \ref{lemma_redundancy_concave} states that the marginal increase of the mean service capacity in the number of redundant executions is decreasing. 

\begin{lemma}
\label{lemma_redundancy_progress}
For all $r \in \mathbb{N}$ and $r \leq M$, $ {g(r,t)} \leq \min\{\Delta  t, rt\}$. 
\end{lemma}
\begin{proof}
As shown in the proof of Lemma \ref{lemma_redundancy_concave}, $g(r,t) = \int_{0}^{\infty} (1-F^r(x,t)) dx$. Therefore, it follow that:
\vspace{-.5em}
\begin{equation}
\begin{split}
\int_{0}^{\infty} (1-F^r(x,t)) dx & = \int_{0}^{\infty} (1-F(x,t))\sum_{l=0}^{r-1}(F^l(x,t)) dx \\  & \geq r \int_{0}^{\infty} (1-F(x,t))F^{r-1}(x,t) dx \\ & = r(g(r,t) - g(r-1,t)).
 \end{split}
 \vspace{-.5em}
\end{equation}
which implies $g(r,t) \leq \frac{r}{r-1}g(r-1,t)$. Hence, $g(r,t) \leq rg(1,t)$. Moreover, we have $g(1,t) = \mathbbm{E}\big[S_i(0,t] \big] = t$. Thus, we have: 
\vspace{-.5em}
\begin{equation}
\label{proof_part_1_lemma_1}
g(n,t) \leq nt.
\vspace{-.5em}
\end{equation} 
Since $g_i(t) = S_i(0,t] \leq \Delta t$, it follows that:
\vspace{-.5em}
\begin{equation}
\label{proof_part_2_lemma_1}
\mathbbm{E}\Big[\max_{i=1,2,\cdots,n}\big\{ g_i(t) \big\}\Big] \leq \Delta t. 
\vspace{-.5em}
\end{equation}
The result follows from \eqref{proof_part_1_lemma_1} and \eqref{proof_part_2_lemma_1}. 
\end{proof}

Lemma \ref{lemma_redundancy_progress} states that the mean service capacity under redundant execution can grow at most linearly in the redundancy, $rt$, and is bounded by the peak service capacity of any single redundant copy, $\Delta t$. 

Given Assumption \ref{assumption_approximation}, the last checkpoint time for job $j$, $t^{L_j }_j$, is also the completion time $c_j$ and satisfies the following equation:
 \vspace{-.5em}
\begin{equation}
\label{job_completion_time_speedup}
\sum_{k=1}^{L_j} g(r_j^{k},t_j^{k} - t_j^{k-1}) = p_j.
 \vspace{-.5em}
\end{equation}

In the sequel, we shall also make use of the speedup function, $h_j(\bm{t_j},\bm{r_j},t)$, defined as follows:
\vspace{-.5em}
\begin{equation}
h_j(\bm{t_j},\bm{r_j},t) = \left\{\begin{array}{cc}
\frac{g(r_j^{k}, t_j^{k} - t_j^{k-1})}{t_j^{k} - t_j^{k-1}} & \ t \in (t_j^{k-1},t_j^{k}], \\
0 & \mbox{otherwise}.
\end{array}\right.
\vspace{-.5em}
\end{equation}
The speedup function captures the speedup that redundant execution is delivering in a checkpointing interval relative to a job execution on a unit speed machine. \eqref{job_completion_time_speedup} can be reformulated in terms of the speedup as follows:
 \vspace{-.5em}
\begin{equation}
\label{job_completion_time}
\int_{a_j}^{c_j} h_j(\bm{t_j},\bm{r_j},\tau) d\tau = p_j.
 \vspace{-.5em}
\end{equation}





\begin{remark}
Note that the speedup depends, not only on the number of redundant copies being executed, but also, on all the times when checkpointing occurs. In this sense, $h_j(\bm{t_j},\bm{r_j},t)$ is not a causal function. However, in the following sections, $h_j(\bm{t_j},\bm{r_j},t)$ will be a convenient notation to study competitive performance bounds for our proposed algorithms. 
\end{remark}

\subsection{Job processing in a Multitasking cluster}
\label{resource_sharing}
With multitasking, a server can run several jobs simultaneously and the service a job receives on a server is proportional to the fraction of processing resource it is assigned. 

We will model a cluster allowing multitasking as follows. Comparing with the service model in Subsection \ref{non_resource_sharing}, we include another variable $x_j^k$, to characterize the fraction of resource assigned to job $j$ in the $k$th subinterval, i.e., $(t^{k-1}_j,t^{k}_j]$. Here, we assume that job $j$ shares the same fraction of processing resource on all the machines on which it is being executed. Let $\bm{x_j} = \big(x^k_j \big|k=1,2,\cdots,L_j \big)$ and we define another speedup function, $\hat{h}_j(\bm{t_j},\bm{x_j},\bm{r_j},t)$, as follows:
\begin{equation}
\hat{h}_j(\bm{t_j},\bm{x_j},\bm{r_j},t) = \left\{\begin{array}{cc}
x^{k}_j \cdot h_j(\bm{t_j},\bm{r_j},t) & \ t \in (t_j^{k-1},t_j^{k}]\\
0 & \mbox{otherwise}
\end{array}\right.
\end{equation}

Paralleling \eqref{job_completion_time}, the completion time of job $j$, $c_j$ must satisfy the following equation:
\begin{equation}
\label{job_completion_time_sharing}
\int_{a_j}^{c_j} \hat{h}_j(\bm{t_j},\bm{x_j},\bm{r_j},\tau) d\tau = p_j 
\end{equation}

In the sequel, we will design and analyze algorithms under both the multitasking mode and the non-multitasking mode.

\subsection{Competitive Performance Metrics}
In this paper, we will study algorithms for scheduling, which involves determining checkpointing times, the number of redundant copies for jobs between successive checkpoints and in the multitasking setting the fraction of resource shares. Note that, when there is no variability in the machine's service capacity, our problem reduces to job scheduling on multiple unit-speed processors with the objective of minimizing the overall flowtime. This has been proven to be NP-hard even when preemption and migration are allowed and previous work \cite{speed,SRPT-redundancy-sharing} has adopted a resource augmentation analysis. Under such analysis, the performance of the optimal algorithm on $M$ unit-speed machines is compared with that of the proposed algorithms on $M$ $\delta$-speed machines where $\delta > 1$.  

The following definition characterizes the competitive performance of an online algorithm using resource augmentation. 

\begin{definition}
\cite{speed} An online algorithm is \textit{$\delta$-speed c-competitive} if the algorithm's objective is within a factor of $c$ of the optimal solution's objective when the algorithm is given $\delta$ resource augmentation.
\end{definition}

In this paper, we also adopt the resource augmentation setup to  bound the competitive performance of our proposed algorithms. With resource augmentation, the service capacity in each checkpointing interval under our algorithms is scaled by $\delta$.  Similarly, the value of the speedup functons, i.e., $h_j(\bm{t_j},\bm{r_j},t)$ and $\hat{h}_j(\bm{t_j},\bm{x_j},\bm{r_j},t)$, under our algorithms is $\delta$ times that under the optimal algorithm of the same variables.

\section{Algorithm Design in a Non-Multitasking Cluster}
\label{no-resource-sharing}
In a non-multitasking cluster, each server can only serve one job at any time.  Before going to the details of algorithm design, we first state the optimal problem formulation. For ease of illustration, we let ${\bm{y}_j} = (\bm{t_j},\bm{r_j},L_j)$ denote the checkpointing trajectory of job $j$ and $\bm{y} = (\bm{y}_j|j=1,2,\cdots,N)$ that for all jobs. Moreover, let $\mathbbm{1}({A})$ denote the indicator function that takes value 1 if $A$ is true and 0 otherwise. The optimal problem formulation is as follows:
\begin{align*}
  \min_{{\bm{y}}} &  \  \sum_{j=1}^{N}  (c_j - a_j)  \tag{OPT}  
  &\\
   \mbox{such that} 
  & \ \mbox{ (a), \ (b), \ (c), \  (d) are satisfied}
\end{align*}
\begin{itemize}
\item[(a)] Job completion: The completion time of job $j$, $c_j$, satisfies: $\int_{a_j}^{c_j}h_j(\bm{t_j},\bm{r_j},t)dt = p_j, \ \ \forall j$. 
\item[(b)] Resource constraint: The total number of redundant executions at any time $t \geq 0$ is no larger than the number of machines, $M$, i.e., $\sum_{j: a_j \leq t} \sum_{k=1}^{L_j}r_j^{k} \cdot \mathbbm{1}({t \in (t^{k-1}_j,t^{k}_j])} \leq M, \ \ \forall t$.  
\item[(c)] Checkpoint trajectory: The number of checkpoints for each job is between 2 and $2N$ since there are $2N$ job arrivals and departures, i.e., $L_j \in \{1,2,\cdots,2N-1\}$. The checkpoint times of job $j$, $\bm{t}_j$, satisfy: $\bm{t}_j \in \mathcal{T}_j^{L_j+1}$ where $\mathcal{T}_j^{L_j+1} = \big \{(t_0,t_1,\cdots,t_{L_j}) \in \mathbb{R}^{L_j+1}|a_j = t_0 < t_1 < \cdots < t_{L_j} = c_j \big\}$. Moreover, the number of redundant copies must be an integer, i.e., $\bm{r}_j \in \mathbb{N}^{L_j}$. 
\item[(d)] Checkpointing overhead constraint: Job checkpoints must satisfy Assumption 1, i.e., for $0 \leq k \leq L_j$, $t^k_j \in \{a_j\}_j \cup \{c_j\}_j$.
\end{itemize}

\vspace{.5em}
Since the OPT problem is NP-Hard, we propose to design a heuristic to schedule redundant jobs, i.e., SRPT+R, which is a simple extension of the SRPT scheduler \cite{SRPT}.

\subsection{SRPT+R Algorithm and the performance guarantee} 
Let $p_j(t)$ denote the amount of the unprocessed work for job $j$ at time $t$ and $n(t)$ denote the number of active jobs at time $t$. In this section, we will assume without loss of generality that jobs have been ordered such that $p_1(t) \leq p_2(t) \leq \cdots \leq p_{n(t)}(t)$. 

At a high level, the algorithm works as follows. When $n(t) \geq M$, the $M$ jobs with smallest $p_j(t)$, i.e., Job 1 to $M$ are each assigned to a server while the others wait. If $n(t) < M$, the job with the smallest $p_j(t)$, i.e., Job 1, is scheduled on $M- \lfloor\frac{M}{n(t)}\rfloor (n(t)-1)$ machines and the others are scheduled on $\lfloor\frac{M}{n(t)}\rfloor$ machines each. Here, $\lfloor x \rfloor$ represents the largest integer which does not exceed $x$. 

The corresponding pseudo-code is exhibited as Algorithm \ref{SRPT+R}. 
\begin{algorithm}[!h]
\label{SRPT+R}
\caption{SRPT+R Algorithm}
\While{A job arrives at or departure from the system}
{
Sort the jobs in the order such that $p_1(t) \leq p_2(t) \leq \cdots \leq p_{n(t)}(t)$ and count the number of redundant copies being executed for job $j$, $r_j$ \;
Initialize $M(t)$ to be the set of idle machines \;
\If{$n(t) < M$}
{
\For{$j=1,2,\cdots,n(t)$}
{\If{$j=1$}
{
  $r_j(t) = M - (n(t)-1)\lfloor \frac{M}{n(t)} \rfloor$\;
}
\Else{$r_j(t) = \lfloor \frac{M}{n(t)} \rfloor$\;}
Checkpoint job $j$ and assigns its redundant executions to $r_j(t)$ machines which are uniformly chosen at random from $\{1,2,\cdots,M\}$\;
}
}
\If{$n(t) \geq M$}
{\For{$j=1,2,\cdots,n(t)$}
{\If{$j \leq M$}{Checkpoint job $j$ and assign it to one machine which is uniformly chosen at random from $\{1,2,\cdots,M\}$\;}
\Else{Checkpoint job $j$\;}
}
}


}
 \vspace{-.5em}
\end{algorithm}
\vspace{-.2em}

Our main result, characterizing the competitive performance of SRPT+R, is given in the following theorem:  

\begin{theorem}
\label{theorem_1}
SRPT+R is $(1+\epsilon)$-speed $O(\frac{1}{\epsilon})$-competitive with respect to the total job flowtime. 
\end{theorem}

We will prove Theorem \ref{theorem_1} by adopting the online dual fitting approach. The first step is to formulate a minimization problem which serves as an approximation to the optimal cost, $OPT$ with a guarantee that the cost of the approximation is within a constant of $OPT$. We then formulate the dual problem for the approximation and exploit the fact that a feasible solution to this dual problem gives a lower bound on its cost, which in turn is a constant times the cost of the proposed algorithm.

\begin{remark}
 It is worth to note that, when there is no machine service variability, SRPT+R performs exactly the same as the traditional SRPT algorithm on multiple machines. As a result, our proposed dual fitting framework can also show that SRPT is $(1+\epsilon)$-speed, $(3 + \frac{3}{\epsilon})$ competitive with respect to the overall job flowtime. When given small resource augmentation where $\epsilon \leq \frac{1}{3}$, our result improves the recent result in \cite{SRPT}, which states, SRPT on multiple identical machines is $(1+\epsilon)$-speed, $\frac{4}{\epsilon}$-competitive in terms of the overall job flowtime. 
\end{remark}

\vspace{-.5em}

\subsection{Proof of Theorem \ref{theorem_1}}
\label{proof_dual_fitting_technique}
To prove Theorem \ref{theorem_1}, we shall first both approximate the objective of OPT and relax Constraint (d) in OPT to obtain the following problem P1:
\vspace{-.5em}
\begin{align*}
  \min_{{\bm{y}}} &  \  \sum_{j=1}^{N}  \int_{a_j}^{\infty}\frac{({t-a_j + 2{p_j}})}{p_j} \cdot h_j(\bm{t_j},\bm{r_j},t)dt  \tag{P1} &\\
   \mbox{s.t.} 
  &  \  \int_{a_j}^{\infty}h_j(\bm{t_j},\bm{r_j},t)dt \geq p_j, \ \ \forall j, &\\
  &  \  \sum_{j: a_j \leq t} \sum_{k=1}^{L_j}r_j^{k} \cdot \mathbbm{1}({t \in (t^{k-1}_j,t^{k}_j])} \leq M, \ \ \forall t, &\\
  &   \  L_j \in \{1,2,\cdots,2N-1\},  \ \ \bm{t}_j \in \mathcal{T}_j^{L_j+1}, \ \ \bm{r}_j \in \mathbb{N}^{L_j}, \ \forall j.
  \vspace{-.5em}
\end{align*}


Let $OPT$ denote the cost, i.e., the overall job flowtime, achieved by an optimal scheduling policy. The following lemma guarantees that the optimal cost of P1, denoted by $P1$, is not far from $OPT$.

\begin{lemma}
\label{lemma_approximation}
$P1$ is upper bounded by $\big(1+ 2\Delta \big)\cdot OPT$, i.e., $P1 \leq \big(1+ 2\Delta \big)\cdot OPT$. 
\end{lemma}


Let $\alpha_j$ and ${\beta(t)}$ denote the Lagrangian dual variables corresponding to the first and second constraint in P1 respectively. Define $\bm{\alpha} = (\alpha_j \big|j = 1,2,\cdots,N)$ and $\bm{\beta} = (\beta(t)|t \in \mathbb{R^+})$. The Lagrangian function associated with P1 can be written as:
\vspace{-.2em}
\begin{equation*}
\begin{split}
\Phi({\bm{y}},\bm{\alpha}, \bm{\beta}) &  =   \sum_{j=1}^{N}  \int_{a_j}^{\infty}\frac{({t-a_j + 2{p_j}})}{p_j} \cdot h_j(\bm{t_j},\bm{r_j},t)dt \\  & + \int_{0}^{\infty} \beta(t) \big(\sum_{j: a_j \leq t} \sum_{k=1}^{L_j}r_j^{k}  \mathbbm{1}({t \in (t^{k-1}_j,t^{k}_j]}) - M \big)dt  \\ & - \sum_{j=1}^{N}\alpha_j \big (\int_{a_j}^{\infty}h_j(\bm{t_j},\bm{r_j},t)dt  - p_j \big),
 \end{split}
  \label{resource_large_combine}
  \vspace{-.2em}
\end{equation*}
with the dual problem for P1 given by:
\begin{align*}
  \max_{\bm{\alpha} \geq \bm{0}, \bm{\beta} \geq \bm{0}} \min_{\bm{\bm{y}} } & \quad \Phi(\bm{y}, \bm{\alpha}, \bm{\beta})  \tag{D1} &\\
   \mbox{s.t.} 
  &  \ L_j \in \{1,2,\cdots,2N-1\}, \  \bm{r}_j \in \mathbb{N}^{L_j}, \  \bm{t}_j \in \mathcal{T}_j^{L_j+1}.
\end{align*}

Applying weak duality theory for continuous programs \cite{continuous_duality_theory}, we can conclude that the optimal value to D1 is a lower bound for $P1$. Moreover, the objective of D1 can be reformulated as shown in \eqref{transform-dual}.
\newcounter{mytempeqncnt}
\begin{figure*}[ht]
\normalsize
\setcounter{equation}{6}
\begin{equation}
\Phi(\bm{y}, \bm{\alpha}, \bm{\beta}) = \sum_{j}\alpha_j p_j - M\int_{0}^{\infty}\beta(t)dt  +   \int_{0}^{\infty} \sum_{j: a_j \leq t} \Big[(\frac{t-a_j}{p_j} + 2- \alpha_j)h_j(\bm{t_j},\bm{r_j},t) +  \beta(t)\sum_{k=1}^{L_j}r_j^{k} \cdot \mathbbm{1}({t \in (t^{k-1}_j,t^{k}_j])}\Big]dt.
\label{transform-dual}
\end{equation}
\hrulefill
\end{figure*}

Still it is difficult to solve D1 as it involves a minimization of a complex objective function of integer valued variables. However, it follows from Lemma \ref{lemma_redundancy_progress} that $r_j^{k} \geq \mathbbm{1}({t \in (t^{k-1}_j,t^{k}_j]}) \cdot h_j(\bm{t_j},\bm{r_j},t)$ for all $j$ and $t \geq a_j$, thus, we have that, 
\begin{equation*}
\vspace{-.2em}
\label{lemma_2_application_1}
\begin{split}
\sum_{k=1}^{L_j}r_j^{k} \mathbbm{1}({t \in (t^{k-1}_j,t^{k}_j])} & \geq \sum_{k=1}^{L_j} \mathbbm{1}({t \in (t^{k-1}_j,t^{k}_j])} h_j(\bm{t_j},\bm{r_j},t) \\ & = h_j(\bm{t_j},\bm{r_j},t).
\end{split}
\vspace{-.2em}
\end{equation*}

Therefore, it can be readily shown that the second term in the R.H.S of $\Phi(\bm{y}, \bm{\alpha}, \bm{\beta})$ in \eqref{transform-dual} is lower bounded by: 
\begin{equation*}
\int_{0}^{\infty} \sum_{j: a_j \leq t} \Big[\Big(\frac{t-a_j}{p_j} + 2- \alpha_j + \beta(t)\Big)\cdot h_j(\bm{t_j},\bm{r_j},t) \Big]dt.
\end{equation*} 
As a result, for a fixed $\alpha_j$ and $\beta(t)$ such that for all $t \geq a_j$
\begin{equation}
\label{optimization-approximation-dual-variables}
\frac{t-a_j}{p_j} + 2 - \alpha_j + \beta(t) \geq 0,
\end{equation} 
the minimum of $\Phi(\bm{y}, \bm{\alpha}, \bm{\beta})$ can be attained by setting all $\bm{r_j}$ to $\bm{0}$ and $\bm{t_j} = (a_j,c_j)$. In this solution, there are no other checkpoints for job $j$ other than the job arrival and departure.  

Therefore, restricting $\bm{\alpha}$ and $\bm{\beta}$ to satisfy \eqref{optimization-approximation-dual-variables} would give a lower bound on D1 and results in the following optimization problem: 
\vspace{-.3em}
\begin{align*}
  \max_{\bm{\alpha}, \bm{\beta}} & \qquad \qquad  \sum_{j}\alpha_j p_j - M\int_{0}^{\infty}\beta(t)dt \tag{P2} &\\
   s.t. & \qquad \qquad  \alpha_j - \beta(t) \leq \frac{t-a_j}{p_j} + 2, \ \ \forall j,\ \ t \geq a_j, &\\
  &  \qquad \qquad \alpha_j \geq 0, \ \ \forall j \\
  &  \qquad \qquad \beta(t) \geq 0, \ \ \forall t 
\end{align*}

Based on Lemma \ref{lemma_approximation}, we conclude that  $ P2 \leq P1 \leq \big(1+2\Delta \big)\cdot OPT$ where $P2$ is the optimal cost for P2. 

Next, we shall find a setting of the dual variables in P2 such that the corresponding objective is lower bounded by $O(\epsilon)\cdot SR$ under a $(1+\epsilon)$-speed resource augmentation. To achieve this, we first consider a pure SRPT scheduling process that does not exploit job redundancy. We then use this to motivate a setting of dual variables which feasible for P2.  Finally, we show that the objective for this setting of dual variables is at least $O(\epsilon)$ times the cost of SRPT, which is also lower bounded by $O(\epsilon)\cdot SR$ since the cost of SRPT is no smaller than $SR$.

\subsubsection{Setting of dual variables}
\label{setting_dual_variable}
Observe that SRPT+R and SRPT only differ when $n(t) < M$ and that when this is the case SRPT only assigns a single machine to each active job. Since SRPT+R maintains the same scheduling order and each job is scheduled with at least the same number of copies as SRPT, we conclude that the cost of SRPT, denoted by $SRPT$, is a lower bound for $SR$, where $SR$ denotes the overall job flowtime achieved SRPT+R.

In this section, we let $n(t)$ and $p_j(t)$ denote the number of active jobs and the size of the remaining workload of job $j$ under SRPT respectively. 




Let $\Theta_j = \{k: a_k \leq a_j \leq c_k\}$, the set of jobs that are active when job $j$ arrives and $A_j = \{k \neq j: k \in \Theta_j \ \mbox{and} \  p_k(a_j) \leq p_j\}$, i.e., jobs whose residual processing time upon job $j$'s arrival is less than job $j$'s processing requirement. Define $\rho_j = |A_j|$, we shall set the dual variables as follows:
\vspace{-.5em}
\begin{equation}
\label{dual_variable_1}
\begin{split}
\alpha_j & = \frac{1}{(1+\epsilon)p_j}\sum_{k=1}^{\rho_j}\Big( \big \lfloor \frac{n(a_j) - k}{M}\big\rfloor - \big \lfloor\frac{n(a_j) - k - 1}{M}\big \rfloor \Big) p_k(a_j) \\ & + \frac{1}{1+\epsilon} \Big( \big \lfloor \frac{n(a_j) - \rho_j - 1}{M}\big \rfloor + 1\Big),
\end{split}
\vspace{-.5em}
\end{equation}
where $\epsilon > 0$ and 
\vspace{-.5em}
\begin{equation}
\label{dual_variable_2}
\beta(t) = \frac{1}{(1+\epsilon)M}n(t).
\vspace{-.5em}
\end{equation}
We show in the following lemma that this setting of dual variables is feasible.

\begin{figure}
\centering
\includegraphics[width=0.48\textwidth]{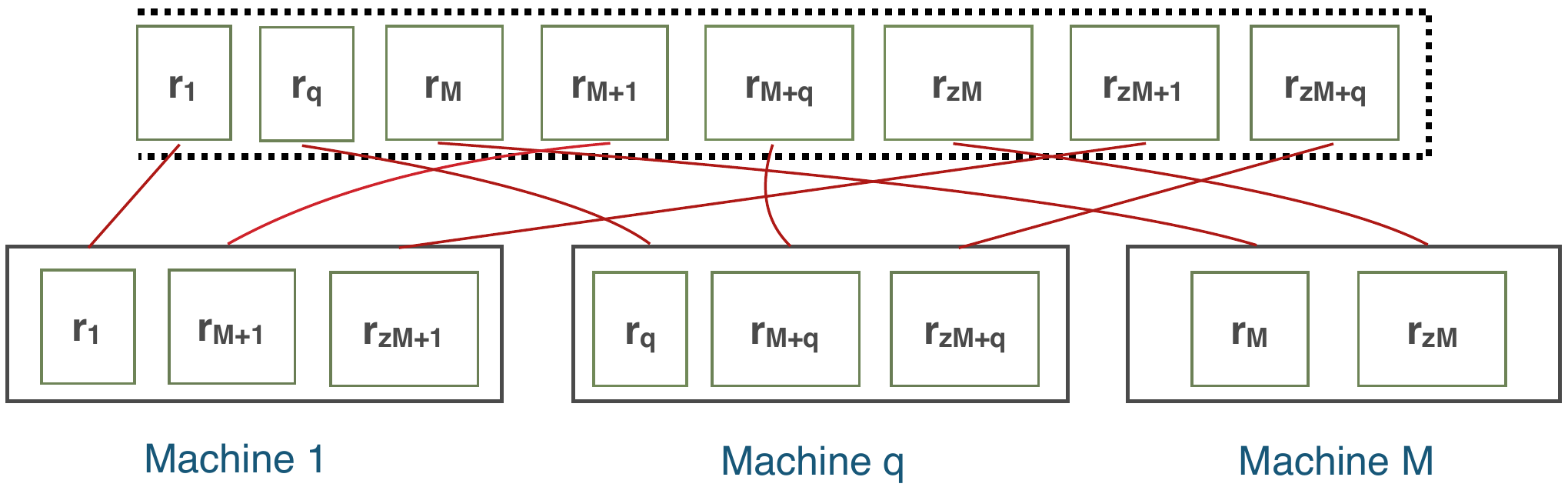}
\caption{The scheduling process of SRPT at time $a_j$ where $n(a_j) = zM + q$ and there are no further job arrivals after $a_j$. Jobs are sorted based on the remaining size, which is denoted by $r_j$ for job $j$, i.e., $r_j = p_j(a_j)$. Jobs indexed by $kM + i$ for some integer valued $k$ and $i$ are assigned to machine $i$. }
\label{assignment}
\end{figure}

\begin{lemma}
The setting of dual variables in \eqref{dual_variable_1} and \eqref{dual_variable_2} is feasible to P2. 
\end{lemma}
\begin{proof}
Since $\bm{\alpha}$ and $\bm{\beta}$ are both nonnegative, it only remains to show $\alpha_j - \beta(t) \leq \frac{t-a_j}{p_j}+  2$ for all $j$ and $t \geq a_j$. First, $\alpha_j$ can be represented as follows: 
\vspace{-.5em}
\begin{equation}
\label{case_1_decom}
\alpha_j  = \frac{\sum_{k=0}^{z}p_{kM + q}(a_j) \mathbbm{1}({kM+q \leq \rho_j})}{(1+\epsilon)p_j}  + \frac{\Big( \big \lfloor \frac{n(a_j) - \rho_j - 1}{M}\big \rfloor + 1\Big)}{1+\epsilon}. 
\end{equation}
For ease of illustration, let $\Omega_1$ and $\Omega_2$  denote the two terms on the R.H.S of \eqref{case_1_decom} respectively.

If $n(a_j) \leq M$, we have $\alpha_j = \frac{1}{1+\epsilon}$ and the result follows. Therefore, we only consider $n(a_j) = zM + q > M$ and analyze the following three cases:

\textbf{Case I:} All the jobs in $\Theta_j$ have completed at time $t$. As depicted in Fig.~\ref{assignment}, if there are no job arrivals after time $a_j$, then, jobs indexed by $km + q$ where $k$ is a non-negative integer are all processed on Machine $q$. Since the service capacity of Machine $q$ is $(t-a_j)$ during $(a_j,t]$, thus, it follows that, 
\vspace{-.5em}
\begin{equation}
\label{case_1_proof}
t - a_j \geq \frac{1}{1+\epsilon} \sum_{k=0}^{z} p_{kM+q}(a_j).
\vspace{-.5em}
\end{equation}
In contrast, if there are other job arrivals after time $a_j$, Machine $q$ needs to process an amount of work which exceeds $\sum_{k=0}^{z} p_{kM+q}(a_j)$, therefore, \eqref{case_1_proof} still holds.  Thus, we have that,
\vspace{-.1em}
\begin{equation}
\begin{split}
\frac{t-a_j}{p_j} - \Omega_1 & \geq \frac{\sum_{k=0}^{z}p_{kM + q}(a_j) \mathbbm{1}({kM + q \geq \rho_j + 1})}{(1+\epsilon)p_j} \\ & \geq \frac{\sum_{k=0}^{z}\mathbbm{1}({kM+q \geq \rho_j + 1})}{1+\epsilon}= \Omega_2.
\end{split}
\vspace{-.5em}
\end{equation}

\textbf{Case II:} The jobs indexed from $1$ to $\kappa$ in $\Theta_j$ have completed and $\kappa \leq \rho_j$. Let $\kappa = z_1M + q_1$. Similar to Case I, it follows that,  
\vspace{-.5em}
\begin{equation}
\label{case_ii_decom}
t - a_j \geq \frac{1}{1+\epsilon} \sum_{k=0}^{z_1} p_{kM + q_1}(a_j). 
\vspace{-.5em}
\end{equation}
In addition, the number of active jobs, $n(t)$, is no less than $n(a_j) - \kappa$. Therefore, we have:  
\vspace{-.5em}
\begin{equation}
\begin{split}
\alpha_j & \overset{\mbox{(ii)}} \leq \frac{\sum_{k=0}^{z_1} p_{kM + q_1}(a_j) + \frac{\lceil\frac{\rho_j-\kappa}{M}\rceil}{1+\epsilon}}{(1+\epsilon)p_j}  + \frac{\Big( \big \lfloor \frac{n(a_j) - \rho_j - 1}{M}\big \rfloor + 1\Big)}{1+\epsilon} \\ & \overset{\mbox{(iii)}} \leq \frac{t-a_j}{p_j} + \frac{1}{1+\epsilon}\lceil\frac{\rho_j-\kappa}{M}\rceil + \frac{\Big( \big \lfloor \frac{n(a_j) - \rho_j - 1}{M}\big \rfloor + 1\Big)}{1+\epsilon} \\
& \leq \frac{t-a_j}{p_j} + \frac{1}{1+\epsilon} \big( \lfloor \frac{n(a_j) - \kappa}{M}\rfloor + 2 
\big) \\ & \leq \frac{t-a_j}{p_j} + \beta(t) + 2,
\end{split}
\vspace{-.5em}
\end{equation}
where $\lceil x \rceil$ denotes the smallest integer which is no less than $x$ and $\mbox{(ii)}$ is due to that $\Omega_1 \leq \frac{1}{(1+\epsilon)p_j} \sum_{k=0}^{z_1} p_{kM + q_1}(a_j) + \frac{1}{1+\epsilon}\cdot \lceil\frac{\rho_j-\kappa}{M}\rceil$. $\mbox{(iii)}$ is due to \eqref{case_ii_decom}.

\textbf{Case III:}  The jobs indexed from $1$ to $\kappa$ in $\Theta_j$ have completed and $\kappa > \rho_j$. In this case, \eqref{case_ii_decom} still holds. Moreover, we have that $\sum_{k=0}^{z_1} p_{kM + q_1}(a_j) \geq \Omega_1 + \lfloor \frac{\kappa - \rho_j}{M} \rfloor p_j$. Therefore, it follows that: 
\vspace{-.5em}
\begin{equation}
\begin{split}
\alpha_j & \leq \frac{t-a_j}{p_j} - \frac{1}{1+\epsilon} \lfloor \frac{\kappa - \rho_j}{M}\rfloor + \frac{1}{1+\epsilon}\lceil \frac{n(a_j) - \rho_j}{M} \rceil 
\\ & \leq \frac{t-a_j}{p_j} + \frac{1}{1+\epsilon} \big( \lfloor \frac{n(a_j) - \kappa}{M} \rfloor + 2 
\big) \\ & \leq \frac{t-a_j}{p_j} + \beta(t) + 2.
\end{split}
\vspace{-.5em}
\end{equation}

Thus, we conclude that, for all the three cases above, the constraint between  $\alpha_j$ and $\beta(t)$ is well satisfied. 
\end{proof}

\subsubsection{Performance bound}
To bound the cost of the dual variables which are set in \eqref{dual_variable_1} and \eqref{dual_variable_2}, we first show the following lemma to quantify the total job flowtime under SRPT in the transient case where there are no job arrivals after time $t$. 
\begin{lemma}
\label{srpt_optimality_transient}
When there are no job arrivals after time $t$, the overall remaining job flowtime under SRPT scheduling, $F(t)$, is given by:
\vspace{-.5em}
\begin{equation}
\label{equation_setting_dual}
F(t) = \sum_{j=1}^{n(t)} (\big 
 \lfloor \frac{n(t)-j}{M}\big \rfloor + 1) p_j(t).
 \vspace{-.5em}
\end{equation}
\end{lemma}

\begin{proof}
In this proof, we shall not assume resource augmentation. Let $f_j(t)$ denote the remaining flowtime for job $j$ at time $t$. Thus, the job completion time, $c_j$ is equal to $f_j(t) + t$. Since we have indexed jobs such that $p_1(t) \leq p_2(t) 
\leq \cdots \leq p_{n(t)}(t)$, under SRPT, it follows that $c_1 \leq c_2 \leq \cdots \leq c_{n(t)}$. 
When $n(t) \leq M$, \eqref{equation_setting_dual} follows immediately since all jobs can be scheduled simultaneously and $f_j(t)$ is equal to $p_j(t)$. 

Let us then consider the case where $n(t) > M$. Let $n(t) = zM + q$ where $z \geq 1$, $0 \leq q \leq M-1$ and $z$, $q$ are non-negative integers. We first show that for all $k$ such that $M \leq k \leq n(t)$, the following result holds: 
\vspace{-.4em}
\begin{equation}
\label{equation_lemma_to_prove}
\sum_{j=k-M+1}^{k} f_{j}(t) = \sum_{j=1}^{k}p_j(t).
\vspace{-.1em}
\end{equation}
As illustrated in Fig.~\ref{srpt_waiting}, at any time between $t$ and $c_{1}$, there are $(k-M)$ jobs waiting to be processed among those $k$ jobs which complete first. Hence, the accumulated waiting time in this period is $(k-M)f_1(t)$. Similarly, at any time between $c_{1}$ and $c_{2}$, there are  $(k-M-1)$ jobs waiting to be processed and they contribute $(k-M-1)\cdot(c_{2} - c_{1}) = (k-M-1)\cdot(f_{2}(t) - f_{1}(t))$ waiting time. Hence, the total waiting time of the $k$ jobs is given by: 
\begin{equation}
\sum_{j=0}^{k-M-1}(k-M-j)\cdot(f_{j+1}(t) - f_{j}(t)) = \sum_{j=1}^{k-M}f_{j}(t).
\end{equation}
Therefore, the total remaining flowtime for these $k$ jobs is as follows:  
\begin{equation}
\label{lower_bound_compleiton_time}
\sum_{j=1}^{k}f_{j}(t) = \sum_{j=1}^{k}p_j(t) + \sum_{j=1}^{k-M}f_{j}(t).
\end{equation}
By shifting terms in \eqref{lower_bound_compleiton_time}, we have: $\sum_{j=k-M+1}^{k} f_{j}(t) = \sum_{j=1}^{k}p_j(t)$. Summing up all job flowtime, it follows that:
\begin{equation}
\label{lower_bound_optimal_scheduling}
\begin{split}
 \sum_{j=1}^{n(t)}f_j(t) & = \sum_{j=1}^{q}f_{j}(t) + \sum_{k=1}^{z} \sum_{j=(k-1)M + q+1}^{kM + q}f_{j}(t) \\ &  \overset{(i)} =  \sum_{j=1}^{q}{p_j(t)} + \sum_{k=1}^{z}\sum_{j=1}^{kM+q}p_j(t)  \\ & = \sum_{j=1}^{n(t)} (\big 
 \lfloor \frac{n(t)-j}{M}\big \rfloor + 1) p_j(t),
\end{split}
\end{equation}
where on the R.H.S. of $(i)$, the first term is due to that the flowtime of the first $q$ jobs is equal to their remaining job size and the second term is due to that $\sum_{j=(k-1)M + q+1}^{kM + q}f_{j}(t) = \sum_{j=1}^{kM+q}p_j(t)$. This completes the proof.
\end{proof}

Based on Lemma \ref{srpt_optimality_transient}, if job $j$ never arrive to the system and the subsequent jobs do not enter the system, the overall remaining job flowtime at time $a_j$ is given by: 
\vspace{-.5em}
\begin{equation}
\label{before_arrival_j}
F_j^{'}(a_j) = \sum_{k=1}^{n(a_j) - 1} (\big 
 \lfloor \frac{n(a_j)-1-k}{M}\big \rfloor + 1) p_k(a_j).
 \vspace{-.5em}
\end{equation}
In contrast, when job $j$ arrives and the subsequent jobs do not arrive to the system at time $a_j$, the overall remaining job flowtime at time $a_j$ is as follows:
\vspace{-.5em}
\begin{equation}
\label{after_arrival_j}
\begin{split}
F_j(a_j) & = \sum_{k=1}^{\rho_j} (\big \lfloor \frac{n(a_j) - k}{M}\big\rfloor + 1) p_k(a_j) \\ & + \Big( \big \lfloor \frac{n(a_j) - \rho_j - 1}{M}\big \rfloor + 1\Big)p_j \\ & + \sum_{k=\rho_j+1}^{n(a_j)} (\big \lfloor \frac{n(a_j) - k}{M}\big\rfloor + 1) p_k(a_j). 
\end{split} 
\vspace{-.5em}
\end{equation}
Therefore, one can view $\alpha_j$ as the incremental increase of the overall job flowtime caused by the arrival of job $j$ by taking the difference of \eqref{before_arrival_j} and \eqref{after_arrival_j} and then dividing by $(1+\epsilon)p_j$. Since we are using a $(1+\epsilon)$-speed resource augmentation, thus, $\sum_{j}p_j\alpha_j$ exactly characterizes the overall job flowtime in SRPT, i.e., $\sum_{j}\alpha_j p_j = SRPT$.

Moreover, $\beta(t)$ reflects the loading condition of the cluster in our setting, thus, 
$M\int_{0}^{\infty}\beta(t) = \frac{1}{1+\epsilon}SRPT$. Therefore, we have $\sum_{j}\alpha_j p_j - M\int_{0}^{\infty}\beta(t) dt = \frac{\epsilon}{1+\epsilon}SRPT$. 

Based on Lemma \ref{lemma_approximation}, we conclude that $\frac{\epsilon}{1+\epsilon}SR \leq \frac{\epsilon}{1+\epsilon}SRPT \leq P2 \leq P1 \leq \big(1+ 2\Delta \big)\cdot OPT$. This implies $SR \leq O(\frac{1}{\epsilon})OPT$ and completes the proof of Theorem \ref{theorem_1}. \qed

\begin{figure}
\centering
\includegraphics[width=0.48\textwidth]{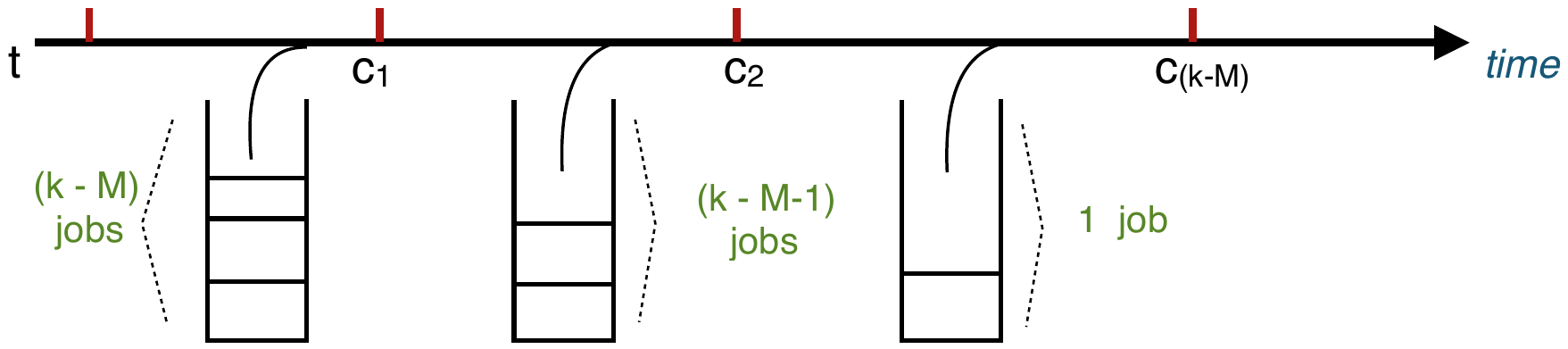}
\caption{The number of jobs waiting to be processed in different time periods where $k > M$.}
\label{srpt_waiting}
\end{figure}

\section{Algorithm Design for Multitasking Processors}
\label{resource-sharing}
In this section, we design scheduling algorithms for clusters supporting multitasking. Besides checkpointing times and level of redundancy, one must introduce additional variables, $\bm{x} = (\bm{x_j}: j = 1,2,\cdots,N)$ where $\bm{x_j} = (x^k_j|k = 1,2,\cdots,L_j)$ are the fractions of resource shares to be allocated to each job during checkpointing intervals. To be specific, we first design the Fair+R Algorithm which is an extension of the Fair Scheduler. Fair+R allows all jobs in the cluster to (near) equally share resources in the cluster, with priority given to those which arrive most recently. We then generalize Fair+R to design the LAPS+R$(\beta)$ algorithm, which is an extension of LAPS (the Latest Arrival Processor Sharing). The main idea of LAPS is to share resources only among a certain fraction of jobs in the cluster \cite{scalably-scheduling}. However, the initial version of LAPS only considers the speed scaling among different jobs, our proposed LAPS+R($\beta$) Algorithm extends this such that redundant copies of jobs can be made dynamically.  In this section, we assume without loss of generality that jobs have been ordered such that $a_1 \leq a_2 \leq \cdots \leq a_{n(t)}$.

\subsection{Fair+R Algorithm and the performance guarantee} 
\label{Fair_R_algorithm}
Let $n(t) = kM + l$ denote the number of jobs which are active in the cluster at time $t$. 

At a high level, Fair+R works as follows. When $n(t) \geq M$, the $kM$ jobs which arrive the most recently, i.e., jobs indexed from $(l+1)$ to ${n(t)}$, are each assigned to one server and gets a resource share of $\frac{1}{k}$. Each server processes $k$ jobs simultaneously. By contrast, if $n(t) < M$, the latest arrival job, i.e., Job ${n(t)}$, is scheduled on $M-\lfloor \frac{M}{n(t)} \rfloor (n(t)-1)$ machines and the others are each scheduled on $\lfloor \frac{M}{n(t)} \rfloor $ machines. In this case, there is no multitasking. 

The corresponding pseudo-code is exhibited in the panel named Algorithm \ref{Fair+R}. Our main result for Fair+R is given in the following theorem: 
\begin{theorem}
\label{theorem_2}
Fair+R is {(4+$\epsilon)$-speed $O(\frac{1}{\epsilon})$-competitive} with respect to the total job flowtime. 
\end{theorem}

\vspace{-.5em}
\begin{algorithm}
\label{Fair+R}
\caption{Fair+R Algorithm}
\While{A job arrives to or departure from the system}
{
Sort the jobs in the order such that $a_1 \leq a_2 \leq \cdots \leq a_{n(t)}$  \;
Compute $n(t) = kM + l$\;
\If{$n(t) \geq M$}
{
\For{$j = l+1, l+2 \cdots, n(t)$}
{$r_{j}(t) = 1$ and 
$x_j(t) = 1/k$\;
}
}
\Else
{
 $r_{n(t)}(t) = M-\lfloor \frac{M}{n(t)} \rfloor (n(t)-1)$ and $x_{n(t)}(t) = 1$\;
 \For{$j = 1,2, \cdots, n(t) - 1$}
{$r_{j}(t) = \lfloor \frac{M}{n(t)} \rfloor $ and 
$x_j(t) = 1$\;
}
}
Checkpoint all jobs and assign job $j$'s redundant executions to $r_j(t)$ machines which are uniformly chosen at random from $\{1,2,\cdots,M\}$ with a resource share of $x_j(t)$;
}
\end{algorithm}
\vspace{-1em}


\subsection{Proof of Theorem \ref{theorem_2}}
Paralleling the proof of Theorem \ref{theorem_1}, we adopt the dual-fitting approach to prove Theorem \ref{theorem_2}. Let $\bm{z}_j = (\bm{t_j},\bm{x_j},\bm{r_j},L_j)$ and $\bm{z} = (\bm{z_j}|j=1,2,\cdots,N)$, we first formulate an approximate optimization problem as follows:
\begin{align*}
  \min_{{\bm{z}}} &  \  \sum_{j=1}^{N}  \int_{a_j}^{\infty}\frac{1}{p_j}({t-a_j + \frac{p_j}{4}}) \widetilde{h}_j(\bm{t_j},\bm{x_j},\bm{r_j},t)dt  \tag{P3} &\\
   s.t. 
  &  \  \int_{a_j}^{\infty}\widetilde{h}_j(\bm{t_j},\bm{x_j},\bm{r_j},t)dt \geq p_j, \ \ \forall j, &\\
  &  \  \sum_{j: a_j \leq t} \sum_{k}x_j^{k}r_j^{k} \cdot \mathbbm{1}{(t \in (t^{k-1}_j,t^{k}_j])} \leq M, \ \ \forall t, &\\
  &   \  L_j \in \{1,2,\cdots,2N-1\},  \ \ \bm{t}_j \in \mathcal{T}_j^{L_j+1}, \ \ \bm{r}_j \in \mathbb{N}^{L_j}, \ \forall j. & \\
  &   \  0 < x_j^{k} \leq 1, \ \ \forall j, 1\leq k \leq L_j. 
\end{align*}

Observe that P3 and P1 differ in both the objective and the second constraint since job $j$ gets a resource share of $x_j^{k}r_j^{k}$ when $t\in(t^{k-1}_{j},t^{k}_j]$.


The dual problem associated with P3 is similar to that of P1 and we only need to modify the first constraint of P2 to yield the following inequality:
\vspace{-.5em}
\begin{equation}
\label{new_dual_constraint}
\alpha_j - \beta(t) \leq \frac{t-a_j}{p_j} + \frac{1}{4} \ \ \forall j;t \geq a_j.
\vspace{-.5em}
\end{equation}
Paralleling Lemma \ref{lemma_approximation}, we have $P4 \leq P3 \leq \big(1+ \frac{1}{4}\Delta \big)\cdot OPT$ where $P4$ and $P3$ are the optimal values of the dual problem and P3 respectively. 


Denote by $A(t)$ the set which contains all jobs that are still active in the cluster at time $t$ under Fair+R. Thus, $n(t) = |A(t)|$.  We shall set $\alpha_j$ as follows:
\vspace{-.5em}
\begin{equation}
\label{choice_alpha}
\alpha_j = \int_{a_j}^{c_j} \alpha_{j}(\tau) d\tau,
\vspace{-.5em}
\end{equation}
 where 
\vspace{-.5em}
\begin{equation}
\label{dudl_variable_alpha}
\begin{split}
\alpha_{j}(t) & = \frac{\sum_{k: a_k \leq a_j} \mathbbm{1}(k \in A(t))  \cdot \mathbbm{1}(n(t) \geq M) \widetilde{h}_k(\bm{t}_k,\bm{x}_{k},\bm{r}_{k},t)}{(4+\epsilon)Mp_j}   \\  & + \frac{\mathbbm{1}({n(t) < M}) \widetilde{h}_j(\bm{t}_j,\bm{x}_j,\bm{r}_{j},t)}{4(4+\epsilon)p_j}. 
 \end{split}
 \vspace{-.5em}
\end{equation}
and the setting of $\beta(t)$ is given by: 
\vspace{-.5em}
\begin{equation}
\label{choice_beta}
\beta(t) = \frac{1}{(4+\epsilon)M} n(t).
\vspace{-.5em}
\end{equation}

Next, we proceed to check the feasibility of these dual variables. Observe that $\alpha_j$ and $\beta(t)$ are nonnegative for all $j,t$ and thus we only need to show they satisfy \eqref{new_dual_constraint}. 
\begin{lemma}
\label{lemma_checking_fair}
The dual variable settings in \eqref{choice_alpha} and \eqref{choice_beta} satisfy the constraint in \eqref{new_dual_constraint}.  
\end{lemma}
\begin{lemma}
\label{lemma_dual_performance}
Under the choice of dual variables in \eqref{choice_alpha} and \eqref{choice_beta},  $\sum_{j=1}^{N}\alpha_j p_j - M\int_{0}^{\infty}\beta(t)dt \geq \frac{\epsilon}{16+4\epsilon}FR$ where $FR$ is the cost of Fair+R.
\end{lemma}

Lemma \ref{lemma_dual_performance} implies that $FR \leq \frac{(16+4\epsilon)P4}{\epsilon}  \leq \frac{16+4\epsilon}{\epsilon} \cdot \big(1+\frac{1}{4}\Delta \big)\cdot OPT = O(\frac{1}{\epsilon})OPT$. This completes the proof of Theorem \ref{theorem_2}. \qed

\subsection{LAPS+R($\beta$) Algorithm and the performance guarantee} 
The algorithm depends on the parameter $\beta \in (0,1)$. Say, $\beta = 1/2$, then the algorithm essentially schedules the $\frac{1}{2}n(t)$ most recently arrived jobs. If there are fewer than $M$ such jobs, they are each assigned an roughly equal number of servers for execution without multitasking. If $\frac{1}{2}n(t) > M$, each job will roughly get a share of $\frac{M}{\frac{1}{2} n(t)}$ on some machine.


For a given number of active jobs $n(t)$, and parameter $\beta$, $z \in \mathbbm{N}$, $\alpha \in \{0,1,\cdots,{M-1}\}$ and $\gamma \in [0,1)$ such that $\beta n(t) =  zM + \alpha + \gamma$.

The LAPS+R($\beta$) Algorithm operates as follows. At time $t$, if $z = 0$, jobs indexed from $({n(t)- \alpha})$ to $({n(t) - 1})$ are scheduled on $\lfloor \frac{M}{\alpha+1} \rfloor $ machines each, and Job ${n(t)}$ is scheduled on the remaining $(M - \alpha \lfloor \frac{M}{\alpha + 1} \rfloor)$ machines. In this case, there is no multitasking. By contrast, if $z\geq 1$, jobs indexed from $({n(t)-zM-\alpha})$ to $({n(t)-1})$ are each assigned a single machine and get a resource share of $\frac{1}{z+1}$. And, Job ${n(t)}$ is scheduled on $(M-\alpha)$ machines with a $\frac{1}{z+1}$ share of its resources. 

The corresponding pseudo-code is exhibited as Algorithm \ref{LAPS+R} in the panel below.

\begin{algorithm}
\label{LAPS+R}
\caption{LAPS+R($\beta$) Algorithm}
\While{A job arrives at or departure from the system}
{
Sort the jobs in the order such that $a_1 \geq a_2 \geq \cdots \geq a_{n(t)}$\;
Compute $\beta n(t) = zM + \alpha + \gamma$ where $\gamma < 1$ and $\alpha < M$\;
\If{$z \geq 1$}{
$r_{n(t)}(t) = (M-\alpha)$ and $x_{n(t)}(t) = \frac{1}{z+1}$\;
\For{$j = n(t) - zM - \alpha, \cdots, n(t) - 1$}
{$r_{j}(t) = 1$ and 
$x_j(t) = \frac{1}{z+1}$\;
}
}
\If{$z < 1$}{
$r_{n(t)}(t) = M-\alpha \lfloor \frac{M}{\alpha + 1} \rfloor$ and $x_{n(t)}(t) = 1$\;
\For{$j = n(t) - \alpha, \cdots, n(t) - 1$}
{$r_{j}(t) = \lfloor \frac{M}{\alpha + 1} \rfloor$ and 
$x_j(t) = 1$\;
}
}
\For{$j = 1,2,\cdots,n(t) - zM - \alpha - 1$}{$x_j(t) = r_j(t)  = 0$\;}
Checkpoint all jobs and assign job $j$'s redundant executions to $r_j(t)$ machines which are uniformly chosen at random from $\{1,2,\cdots,M\}$ with a resource share of $x_j(t)$;
}
 \vspace{-.5em}
\end{algorithm}


\subsubsection{Performance guarantee for LAPS+R($\beta$) and our techniques}
Let $OPT$ and $LR$ denote the cost of the optimal scheduling policy and LAPS+R($\beta$) respectively. The main result in this section, characterizing the competitive performance of LAPS+R($\beta$), is given in the following theorem:
\begin{theorem}
\label{theorem_3}
LAPS+R($\beta$) is {($2 + 2\beta + 2\epsilon)$-speed $O(\frac{1}{\beta \epsilon})$-competitive} 
with respect to the total job flowtime. 
\end{theorem}

The dual fitting approach fails in this setting so we adopt the use of a potential function, which is widely used to derive performance bounds with resource augmentation for online parallel scheduling algorithms e.g., \cite{competitive,scalably-scheduling}. The main idea of this method is to find a potential function which combines the optimal schedule and LAPS+R($\beta$). To be specific, let $LR(t)$ and $OPT(t)$ denote the accumulated job flowtime under LAPS+R($\beta$) with a ($2 + 2\beta + 2\epsilon)$-speed resource augmentation and the optimal schedule, respectively. We define a potential function, $\Lambda(t)$, which satisfies the following properties:
\begin{enumerate}
\item Boundary Condition: $\Lambda(0) = \Lambda(\infty) = 0$.
\item Jumps Condition: the potential function may have jumps only when a job arrives or completes under the LAPS+R($\beta$) schedule, and if present, it must be decreased. 
\item Drift Condition: with a ($2 + 2\beta + 2\epsilon)$-speed resource augmentation, for any time $t$ not corresponding to a jump, and some constant $c$, we have that,
\begin{equation}
\label{sample_bound}
 \frac{d\Lambda(t)}{dt} \leq - {\epsilon \beta} \cdot \frac{dLR(t)}{dt} + {c} \cdot \frac{dOPT(t)}{dt}.
 \end{equation}
\end{enumerate}
By integrating \ref{sample_bound} and accounting for the negative jump and the boundary condition, one can see that the existence of such a potential function guarantees that $LR \leq \frac{c}{\beta \epsilon} OPT$ under a ($2 + 2\beta + 2\epsilon)$-speed resource augmentation.

\subsection{Proof of Theorem \ref{theorem_3}}
To prove Theorem \ref{theorem_3}, we shall propose a potential function, $\Lambda(t)$, which satisfies all the three properties specified above.  

\subsubsection{Defining potential function, $\Lambda(t)$}
Consider a checkpointing trajectory for job $j$ under LAPS+R($\beta$) and the optimal schedule, denoted by $(\bm{t}_j,\bm{x}_j,\bm{r}_j)$ and $(\bm{t}_j^*$,$\bm{x}_j^*,\bm{r}_j^*)$ respectively. Let $\psi^* (t)$ be the jobs that are still active at time $t$ under the optimal scheduling and denote by $\psi(t)$ the set of jobs that are active under LAPS+R($\beta$). Thus, we have that $|\psi(t)| = n(t)$. Further let $n_j(t)$ denote the number of jobs which are active at time $t$ and arrive no later than job $j$ under LAPS+R($\beta$). Define the cumulative service difference between the two schedules for job $j$ at time $t$, i.e., $\pi_j(t)$, as follows:
\begin{equation}
\label{numator_potential}
\pi_j(t) = \max \Big[\int_{a_j}^{t} \widetilde{h}_j(\bm{t}_j^*, \bm{x}_j^*,\bm{r}_j^*,\tau)d\tau - \int_{a_j}^{t} \widetilde{h}_j(\bm{t}_j,\bm{x}_j,\bm{r}_j,\tau)d\tau ~,~ 0 \Big],
\end{equation}
Let $\delta = 2 + 2\beta + 2\epsilon$ and define
\vspace{-.3em}
\begin{equation}
\label{denominator_potential}
f(n_j(t)) = \left\{\begin{array}{cc}
1 & \ \beta n_j(t) \leq M, \\
\frac{M}{\beta n_j(t)} & \mbox{otherwise}.
\end{array}\right.
\vspace{-.3em}
\end{equation}
Note that $f(n_j(t))$ takes the minimum of 1 and $\frac{M}{\beta n_j(t)}$ where the latter is roughly the total resource allocated to job $j$ under LAPS+R($\beta$) if $n_j(t)$ jobs were active at time $t$.

Our potential function is given by:
\vspace{-.3em}
\begin{equation}
\label{definition_potential_funciton_single}
\Lambda (t) =  \sum_{j \text{} \in \psi(t)} \Lambda_j (t),
\vspace{-.3em}
\end{equation}
where $\Lambda_j(t)$ is the ratio between \eqref{numator_potential} and \eqref{denominator_potential}, i.e., 
\vspace{-.3em}
\begin{equation*}
\Lambda_j (t) = \frac{\pi_j(t)}{\delta \cdot f(n_j(t))}.
\end{equation*}

\subsubsection{Changes in $\Lambda(t)$ caused by job arrival and departure}
Clearly, our potential function satisfies the boundary condition. Indeed, since each job is completed under LAPS+R($\beta$), 
thus, $\psi(t)$ will eventually be empty, and $\Lambda (0) = \Lambda(\infty) = 0.$

Let us consider possible jump times. When job $j$ arrives to the system at time $a_j$, $\pi_j(a_j) = 0$ and $f(n_j(t))$ does not change for all $k \neq j$. Therefore, we conclude that the job arrival does not change the potential function $\Lambda(t)$. When a job leaves the system under LAPS+R($\beta$), $f(n_j(t))$ can only increase if job $j$ is active at time $t$, leading to a decrease in $\Lambda_j (t)$. As a consequence, the job arrival or departure does not cause any increase in the potential function, $\Lambda (t)$, thus, the jump condition on the potential function is satisfied.

\subsubsection{Changes of $\Lambda(t)$ caused by job processing}
Beside job arrivals and departures under LAPS+R($\beta$), there are no other events leading to changes in $f(n_j(t))$ and thus changes in $\Lambda_j (t)$ depend only on $\pi_j(t)$, see definition of $\Lambda_j (t)$ in \eqref{definition_potential_funciton_single}. 
Specifically, for all $t \notin \{a_j\}_j \cup \{c_j\}_j$, we have that, $$\frac{d\Lambda (t)}{dt} = \sum_{j  \in \psi(t)} \frac{d \Lambda_j (t)}{dt} = \sum_{j  \in \psi(t)} \frac{{d\pi_j(t)}/{dt}}{\delta \cdot f(n_j(t))},$$ 
where we let $\frac{d\pi_j(t)}{dt} = \lim_{\tau  \rightarrow 0^{+}}\frac{\pi_j(t + \tau) - \pi_j(t)}{\tau}$ and thus $\frac{d\pi_j(t)}{dt}$ exists for all $t \geq 0$. Moreover, we have: 
\begin{equation}
\label{bound_partial_potential}
\begin{split}
\frac{d\pi_j(t)}{dt}  &
\leq \mathbbm{1}(j \in \psi^* (t)){ \widetilde{h}_j(\bm{t}_j^*, \bm{x}_j^*,\bm{r}_j^*,t)} \\ & \quad - \mathbbm{1}(j \notin \psi^* (t)){\widetilde{h}_j(\bm{t}_j, \bm{x}_j,\bm{r}_j,t)},
\end{split}
\end{equation}
indeed, either $j \in \psi^* (t)$ so job $j$ has not completed under the optimal policy and the drift is bounded by the first term in \eqref{bound_partial_potential}, or $j \notin \psi^* (t)$ and the job has completed under the optimal policy, the difference term in \eqref{numator_potential} is positive and its derivative is given by the the second term in \eqref{bound_partial_potential}. Therefore, for all $t \notin \{a_j\}_j \cup \{c_j\}_j$, we have the following upper bound:

\vspace{-.3em}
\begin{equation}
\label{dynamic_potential}
\begin{split}
\frac{d\Lambda (t)}{dt} & \leq \sum_{j  \in \psi(t)} \frac{\mathbbm{1}(j \in \psi^* (t)){ \widetilde{h}_j(\bm{t}_j^*, \bm{x}_j^*,\bm{r}_j^*,t)}  }{\delta \cdot f(n_j(t))}  \\ &  \qquad - \sum_{j  \in \psi(t)}\frac{\mathbbm{1}(j \notin \psi^* (t)) { \widetilde{h}_j(\bm{t}_j, \bm{x}_j,\bm{r}_j,t)} }{\delta \cdot f(n_j(t))} \\ & \leq \underbrace{\sum_{j  \in \psi^* (t)} \frac{{ \widetilde{h}_j(\bm{t}_j^*, \bm{x}_j^*,\bm{r}_j^*,t)} }{\delta \cdot f(n_j(t))}}_{\Gamma^* (t)} \underbrace{- \sum_{j \in \psi(t)\setminus\psi^* (t) } \frac{{ \widetilde{h}_j(\bm{t}_j, \bm{x}_j,\bm{r}_j,t)} }{\delta \cdot f(n_j(t))}}_{\Gamma (t)},
\end{split}
\vspace{-.3em}
\end{equation}
where $\psi(t)\setminus\psi^* (t) $ contains all the jobs that are in $\psi(t)$ but not in $\psi^* (t)$. For ease of illustration, let $\Gamma^* (t)$ and $ \Gamma (t)$ denote the two terms on the R.H.S. of \eqref{dynamic_potential}. In the sequel, we bound these two terms.

\subsubsection{An upper bound of $\Gamma^* (t)$}
When $\beta n_j(t) \geq M$, we have $f(n_j(t)) = M/\beta n_j(t)$, thus, $\frac{\widetilde{h}_j(\bm{t}_j^*,\bm{x}_j^*,\bm{r}_j^*,t)}{f(n_j(t))} = \frac{\widetilde{h}_j(\bm{t}_j^*,\bm{x}_j^*,\bm{r}_j^*,t)}{M/\beta n_j(t)}$. By contrast, when $\beta n_j(t) \leq M$, it follows that $f(n_j(t)) = 1$, so $\frac{\widetilde{h}_j(\bm{t}_j^*,\bm{x}_j^*,\bm{r}_j^*,t)}{f(n_j(t))} = {\widetilde{h}_j(\bm{t}_j^*,\bm{x}_j^*,\bm{r}_j^*,t)}$, which is upper bounded by $\Delta$ based on Lemma \ref{lemma_redundancy_progress}. 

Therefore, we have: $$\frac{\widetilde{h}_j(\bm{t}_j^*,\bm{x}_j^*,\bm{r}_j^*,t)}{\delta f(n_j(t))} \leq  \frac{1}{\delta} \Big(\frac{\widetilde{h}_j(\bm{t}_j^*,\bm{x}_j^*,\bm{r}_j^*,t)}{M/\beta n_j(t)} + \Delta \Big ),$$ and 
\begin{equation}
\label{bound_job_processing_optimal}
\begin{split}
\Gamma^*(t) & \leq \sum_{j \in \psi^* (t)} \frac{1}{\delta} \Big(\frac{\widetilde{h}_j(\bm{t}_j^*,\bm{x}_j^*,\bm{r}_j^*,t)}{M/\beta n_j(t)} + \Delta \Big ) \\ & \leq \frac{\Delta |\psi^* (t)|}{\delta} + \sum_{j \in \psi^* (t)} \beta n(t) \frac{\widetilde{h}_j(\bm{t}_j^*,\bm{x}_j^*,\bm{r}_j^*,t)}{\delta M}   \\ & \leq \Delta |\psi^* (t)|/\delta + \beta n(t)/\delta,
\end{split}
\end{equation}
where the last inequality is due to $$\sum_{j \in \psi^* (t)} {\widetilde{h}_j(\bm{t}_j^*,\bm{x}_j^*,\bm{r}_j^*,t)} \leq \sum_{j \in \psi^* (t)}\sum_{k}x_j^{k}r_j^{k} \cdot \mathbbm{1}{(t \in (t^{k-1}_j,t^{k}_j])}  \leq M,$$ for all $t$. 

\subsubsection{An upper bound of $\Gamma (t)$}
First, $\Gamma (t)$ ban be represented as:
\begin{equation}
\label{bound_first_case_potential_function}
\begin{split}
\Gamma (t) =   \sum_{j \in \psi^* (t) \cap \psi(t)} \frac{\widetilde{h}_j(\bm{t}_j,\bm{x}_j,\bm{r}_{j},t)}{\delta f(n_j(t))} - \sum_{j \in \psi(t)} \frac{\widetilde{h}_j(\bm{t}_j,\bm{x}_j,\bm{r}_{j},t)}{\delta f(n_j(t))}.
\end{split}
\end{equation}
To get an upper bound of $\Gamma (t)$, we will consider two cases. In particular, $\beta n(t) =  zM + \alpha + \gamma$, we consider the case where $z = 0$ and that where $z \geq 1$. 

\vspace{.2em}
\textbf{Case 1:} Suppose $z = 0$, in this case, $\lceil \beta n(t) \rceil \leq M$. Since $n_j(t) \leq n(t)$ for all $1 \leq j \leq n(t)$, it then follows that $\beta n_j(t) \leq M$ and $f(n_j(t)) = 1$, which implies, for all $j \in \psi^* (t) \cap \psi(t)$, $\frac{\widetilde{h}_j(\bm{t}_j,\bm{x}_j,\bm{r}_{j},t)}{\delta f(n_j(t))} \leq \Delta$ since $\widetilde{h}_j(\bm{t}_j,\bm{x}_j,\bm{r}_{j},t) \leq \delta \Delta$ as we are using a $\delta$-speed resource augmentation. Thus, the first term on the R.H.S. of \eqref{bound_first_case_potential_function} is upper bounded by: 
\begin{equation}
\label{bound_first_term_RHS}
\sum_{j \in \psi^* (t) \cap \psi(t)} \frac{\widetilde{h}_j(\bm{t}_j,\bm{x}_j,\bm{r}_{j},t)}{\delta f(n_j(t))} \leq \sum_{j \in \psi^* (t) \cap \psi(t)} \Delta \leq   |\psi^* (t)|\Delta.
\end{equation}

Consider $j \in \psi(t)$ where ${n(t)- \alpha} \leq j \leq n(t)$ and $t \in [t^{k-1}_j,t^k_j)$ for some $k \in \{1,2,\cdots,L_j\}$. Then, the number of redundant executions for job $j$, $r^k_j \geq \lfloor \frac{M}{\alpha+1} \rfloor \geq 1$. Thus, $\widetilde{h}_j(\bm{t}_j,\bm{x}_j,\bm{r}_{j},t) \geq \delta$ and $\frac{\widetilde{h}_j(\bm{t}_j,\bm{x}_j,\bm{r}_{j},t)}{\delta f(n_j(t))} \geq 1$. Combining \eqref{bound_first_case_potential_function} and \eqref{bound_first_term_RHS}, we then have:
\begin{equation}
\label{split_1}
\begin{split}
\Gamma (t) & \leq  \Delta |\psi^* (t)| - (\alpha+1) \leq  \Delta |\psi^* (t)| - \beta n(t),
\end{split}
\end{equation}

\vspace{.2em}
\textbf{Case 2:} Suppose $z \geq 1$, then, $\lceil \beta n(t) \rceil > M$ and $\frac{M}{\beta n(t)} \geq \frac{1}{z+1}$. Similarly, we consider job $j \in \psi(t)$ where $n(t)-kM-\alpha \leq j \leq n(t)$ and $t \in (t_j^{k-1},t_j^{k}]$. Based on the scheduling policy of LAPS+R($\beta$), 
we have that,  $x^k_j = \frac{1}{z+1}$ and $r^k_j \geq 1$. Therefore, $\widetilde{h}_j(\bm{t}_j,\bm{x}_j,\bm{r}_{j},t)$ is bounded by:
\begin{equation}
\label{bounding_h_potential}
\frac{\delta}{z+1} \leq \widetilde{h}_j(\bm{t}_j,\bm{x}_j,\bm{r}_{j},t) \leq \frac{\delta \cdot \Delta}{z+1}. 
\end{equation}

Moreover, we have: $ \min[1~,~\frac{M}{\beta n_j(t)}] \geq \min[1~,~\frac{M}{\beta n(t)}] \geq \frac{1}{z+1}$. Therefore, it follows that:
\begin{equation}
\label{bounding_f_potential}
\frac{1}{z+1} \leq \min \big[1~,~\frac{M}{\beta n_j(t)}\big] \leq  f(n_j(t))   \leq \frac{M}{\beta n_j(t)}. 
\end{equation}
Combining \eqref{bounding_h_potential} and \eqref{bounding_f_potential}, we have that, for all $n(t)-kM-\alpha \leq j \leq n(t) -1$, 
\begin{equation}
\label{term_bound_potential}
 \frac{1/(z+1)}{M/\beta n_j(t)} \leq \frac{\widetilde{h}_j(\bm{t}_j,\bm{x}_j,\bm{r}_{j},t)}{\delta f(n_j(t))} \leq \Delta.
\end{equation} 
Substituting  \eqref{term_bound_potential} into \eqref{bound_first_case_potential_function}, it then follows that, 
\begin{equation}
\label{split_2}
\begin{split}
\Gamma (t)  & \leq \sum_{j \in \psi^* (t) \cap \psi(t)} \Delta - \sum_{j = n(t) - zM - \alpha}^{n(t)}\frac{1/(z+1)}{M/\beta n_j(t)} \\  & \leq \Delta |\psi^* (t)| - \frac{\beta zM (n(t)-\frac{zM}{2}-\frac{\alpha}{2})}{M(z+1)} \\  & \leq \Delta |\psi^* (t)| - \beta(\frac{1}{2}- \frac{\beta}{4})n(t), 
\end{split}
\end{equation}
where the second inequality is due to that $n_j(t) = j$ and the last inequality is because $zM + \alpha \leq \beta n(t)$ and $\frac{z}{z+1} \geq 1/2$. Based on \textbf{Case 1} and \textbf{Case 2}, we have: $\Gamma (t) \leq \Delta |\psi^* (t)| - \beta(\frac{1}{2}- \frac{\beta}{4})n(t)$. Thus, combining \eqref{dynamic_potential} and \eqref{bound_job_processing_optimal}, we have the following upper bound for the drift, $\frac{d\Lambda (t)}{dt}$: 
\begin{equation}
\label{upperbound_A_potential}
\begin{split}
\frac{d\Lambda (t)}{dt} & \leq \Gamma^* (t)  + \Gamma (t) \\ & \leq \Delta |\psi^* (t)|/\delta + \beta n(t)/\delta + \Delta |\psi^* (t)| - \beta(\frac{1}{2}- \frac{\beta}{4})n(t). \\ & = \frac{(\delta+1)\Delta}{\delta}|\psi^* (t)| + \frac{ \beta (1-\delta(\frac{1}{2} - \frac{\beta}{4}))}{\delta}n(t) \\ & \leq \frac{(\delta+1)\Delta}{\delta}|\psi^* (t)| - \frac{\epsilon \beta}{2\delta}n(t), 
\end{split}
\end{equation}
where the last inequality is due to $\delta = 2 + 2\beta + 2\epsilon$ and $\delta(\frac{1}{2}-\frac{\beta}{4}) \geq 1 + \frac{\epsilon}{2}$.

Based on \eqref{upperbound_A_potential}, we then have that, 
\begin{equation}
\label{final_theorem_3}
\begin{split}
 0 & = \Lambda(\infty) - \Lambda(0) \leq \int_{0}^{\infty}\frac{d\Lambda (t)}{dt} dt \\ & \leq
  \frac{(\delta+1)\Delta}{\delta}\int_{0}^{\infty} |\psi^* (t)| dt - \frac{\epsilon \beta}{2\delta} \int_{0}^{\infty} n(t)dt \\ & = \frac{(\delta+1)\Delta}{\delta} OPT - \frac{\epsilon \beta}{2\delta}LR,
\end{split}
\end{equation}
 where the first inequality is due to that there exist negative jumps during the evolution of $\Lambda(t)$.  This completes the proof of Theorem \ref{theorem_3}. \qed

\section{Numerical Studies}
\label{numerical-study}
In this section, we conduct several numerical studies to evaluate our proposed algorithms in both the multi-tasking and non-multi-tasking setting. As pointed out in \cite{trace-archive}, the Gamma distribution is a good fit for the failure model of most parallel and distributed computing systems. Therefore, we apply the Gamma distribution to generate machine service rates in a cluster with 100 machines over a period which lasts 100000 units of time. 

To be more specific, we categorize the service process of each machine into two classes, namely, the Available Period (AP) and Unavailable Period (UP). As depicted in Fig.~\ref{fluctuation}, each UP follows an AP. During an available period, the rate of the machine service capacity is uniformly distributed in $[2,3]$. On the other hand, when the machine is processing jobs in an unavailable period, its rate is uniformly distributed in $[0, 0.3]$. In addition, we apply the statistics of the trace data collected from a computational grid platform (see \cite{trace-archive}) to generate a series of available and unavailable periods for each machine independently. The length of an AP is Gamma distributed with $k=0.34$ and $\theta = 94.35$ where $k$ and $\theta$ are the shape parameter and scale parameter respectively. In contrast, the length of an UP is Gamma distributed with $k = 0.19$ and $\theta = 39.92$. We also normalize all the distributions such that the mean service rate is one.

\begin{figure}
\centering
\includegraphics[width=0.48\textwidth]{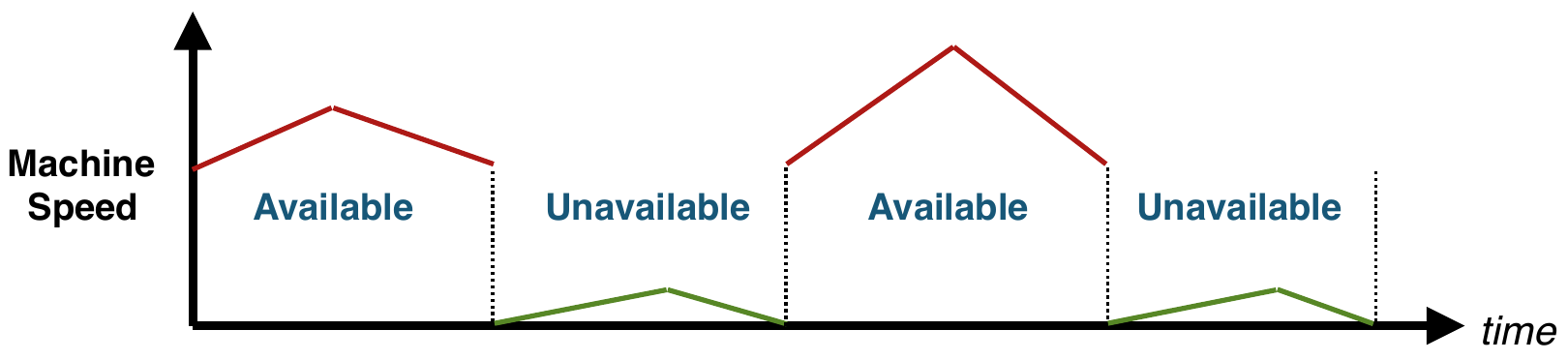}
\vspace{-.8em}
\caption{Fluctuation of machine service rates in different time periods.}
\label{fluctuation}
\vspace{-.6em}
\end{figure} 

\begin{figure}
\centering
\includegraphics[width=0.48\textwidth]{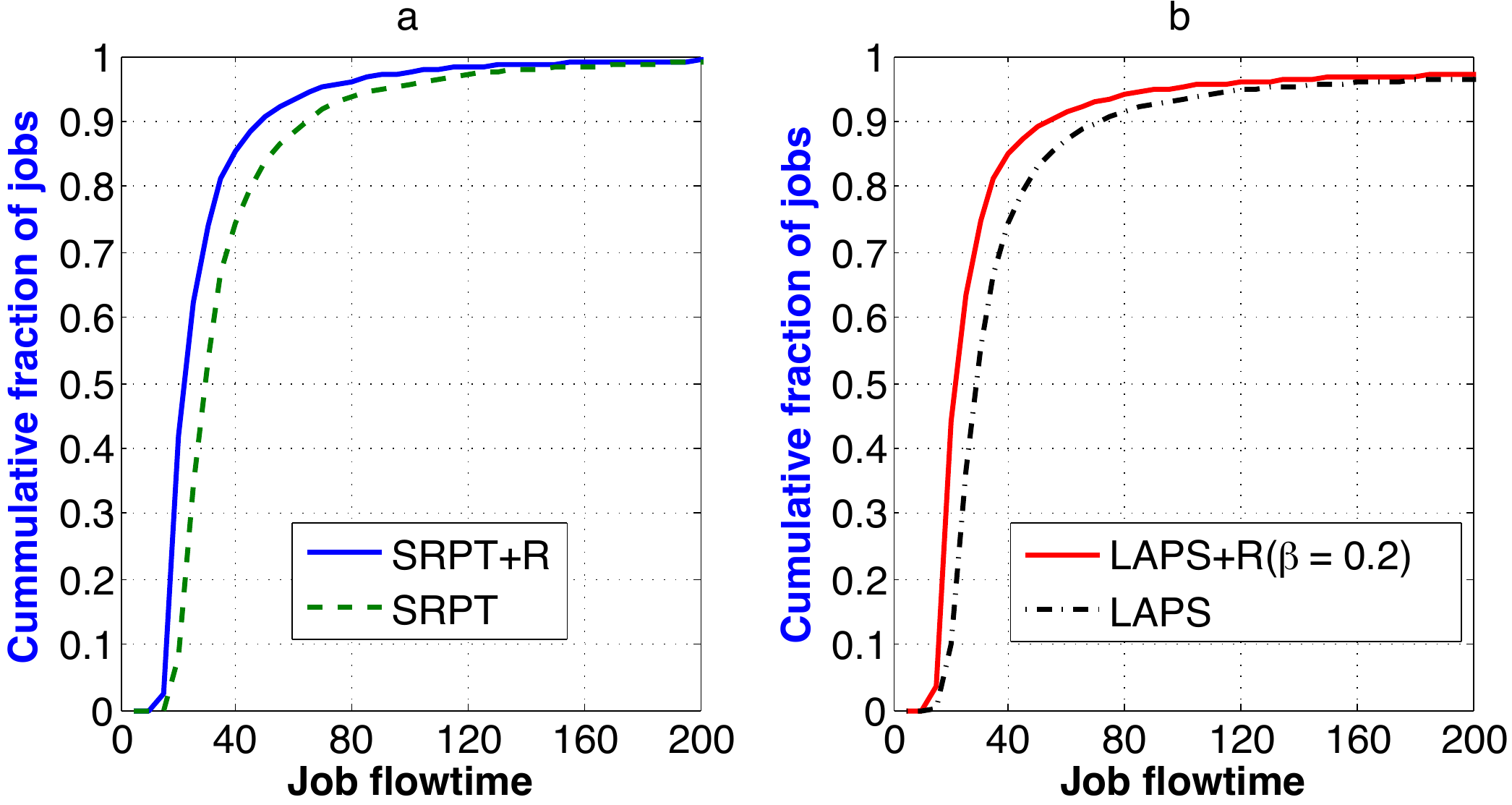}
\vspace{-0.8em}
\caption{The comparison between algorithms with and without redundancy. Panel a shows the CDF of the job flowtimes under both SPRT+R and SRPT. Panel b shows the CDF of the job flowtimes under LAPS+R($\beta$) with $\beta = 0.2$ and LAPS.}
\label{benefit_redundancy}
\vspace{-.6 em}
\end{figure}

In all the following evaluations, we consider time is slotted and the scheduling decisions are made at the beginning of each time slot. The jobs arrive at the cluster following a Poisson Process with rate $\lambda$ and the workload of each job is Pareto distributed as shown below. 
\begin{equation*}
\mbox{P}\{p_j \leq x\} = \left\{\begin{array}{cc}
1-(\frac{b}{x})^{\alpha} & \ x \geq b,\\
0 & \mbox{otherwise,}
\end{array}\right.
\end{equation*}
where $b = 20$ and $\alpha = 2$. It can be readily shown that the mean of the job workload is 40.

In the following simulations, we will compare the average as well as the cumulative distribution function (i.e., CDF) of job flowtimes for different algorithms.

\subsection{Benefit of Redundancy}
In this subsection, we implement scheduling algorithms with both redundancy and no redundancy to characterize the benefit of redundant execution. We set the job arrival rate $\lambda$ to one and depict the simulation results in Fig.~\ref{benefit_redundancy} and Fig.~\ref{average_flow_comparison}. As shown in Fig.~\ref{benefit_redundancy}, more than $85\%$ of jobs can complete within 40 units of time under SRPT+R. As a comparison, only $75\%$ of jobs complete within $40$ units of time under the SRPT scheme. It's worthy noting that this result also applies to LAPS+R($\beta$) and LAPS. Moreover, Fig.~\ref{average_flow_comparison} shows that, with redundancy, the average job flowtime can be reduced by nearly $25\%$ under all the scheduling algorithms.

\begin{figure}
\centering
\includegraphics[width=0.46\textwidth]{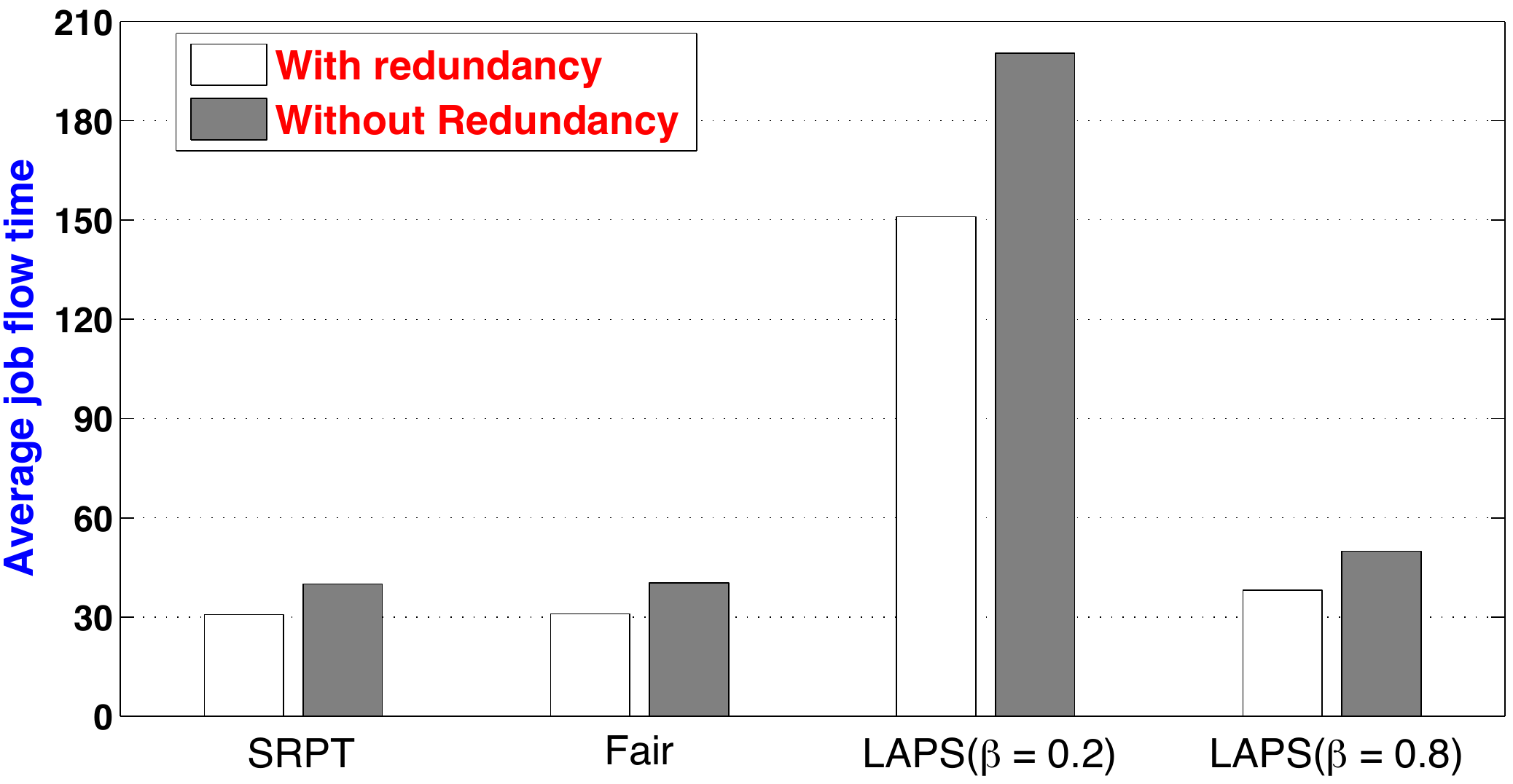}
\vspace{-0.8em}
\caption{Average job flowtime under different scheduling algorithms with and without redundancy.}
\label{average_flow_comparison}
\vspace{-1 em}
\end{figure}

\begin{figure}
\centering
\includegraphics[width=0.46\textwidth]{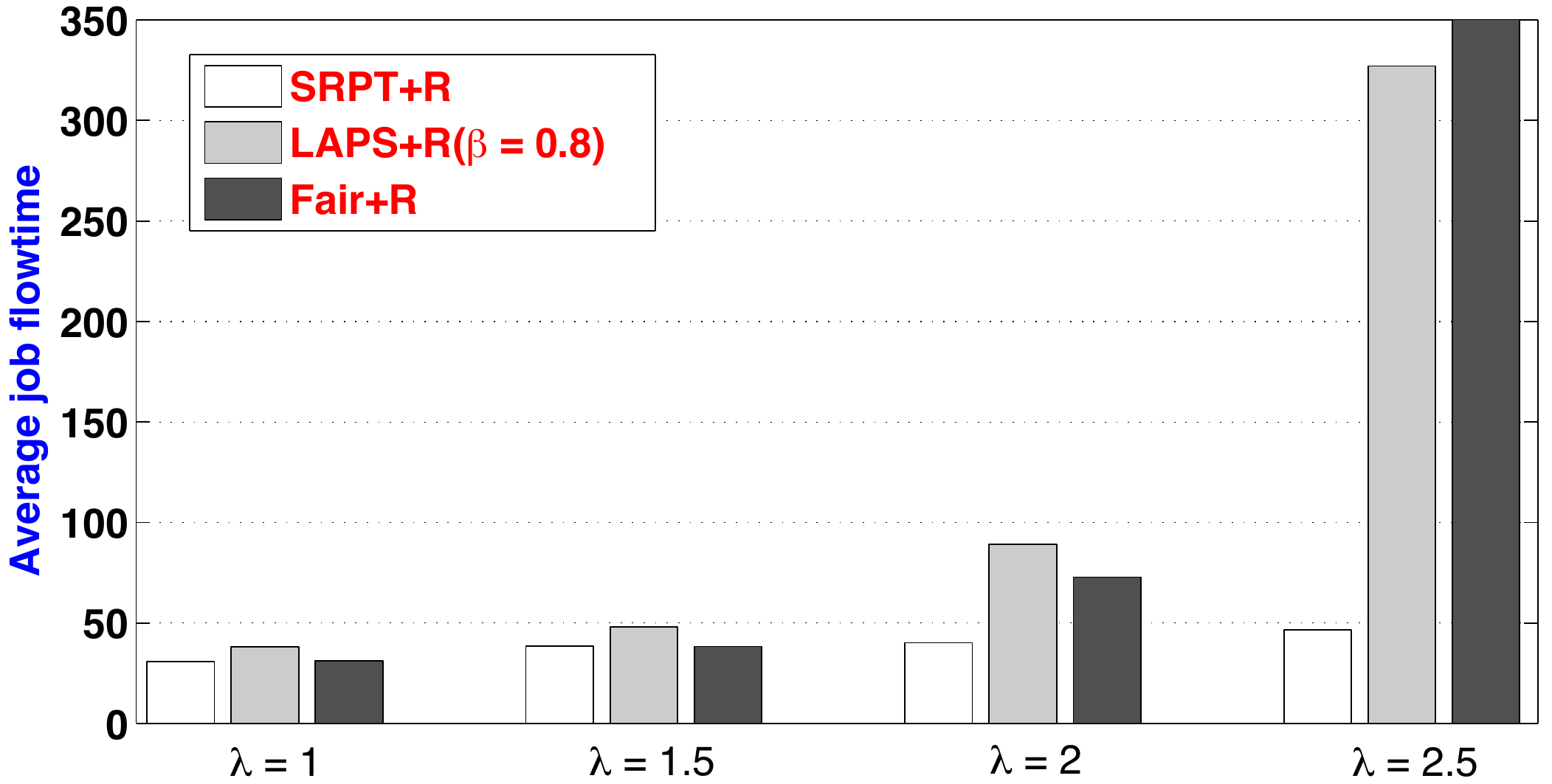}
\vspace{-1.2em}
\caption{Comparison between different algorithms in terms of average job flowtime for different $\lambda$.}
\label{average_flow}
\vspace{-1.2 em}
\end{figure}
\subsection{Comparison of various algorithms}
We conducted a more comprehensive comparison of various algorithms by tuning values of $\lambda$. Following the simulation parameters configured at the beginning of Section \ref{numerical-study}, we can readily show that, $\lambda = 2.5$ reaches the heavy traffic limit above which the system is overloaded. As such, we tune $\lambda$ from 1 to 2.5 in this simulation. Observe in Fig.~\ref{average_flow} that the average job flowtimes under SRPT+R and Fair+R are roughly the same under $\lambda = 1$ and $\lambda = 1.5$. However, when $\lambda$ increases, SRPT+R tends to perform much better than both LAPS+R($\beta$) and Fair+R. For $\lambda = 2$, the average job flowtime under both LAPS+R($\beta = .8$) and Fair+R is two times that under SRPT+R. More importantly, when $\lambda$ hits the heavy traffic limit, the average job flowtime under both LAPS+R($\beta$) and Fair+R increases significantly in $\lambda$ while it doesn't change much under SRPT+R. In addition, Fair+R outperforms LAPS+R($\beta$) when $\lambda$ is below 2. Conversely, if $\lambda$ is above $2$, LAPS+R($\beta$) performs better than Fair+R.

\subsection{The impact of $\beta$ in LAPS+($\beta$)}
Since $\beta$ has a high impact on the performance of LAPS+($\beta$), in this subsection, we tune the values of $\beta$ to illustrate the performance of LAPS+($\beta$) under different settings. 

We depict the comparison results under the heavy traffic regime where $\lambda = 2.5$ in Fig.~\ref{heavy_traffic_laps}. It shows that, when $\beta$ decreases, the number of jobs with small flowtime (less than 200 units of time) increases. Therefore, small jobs benefit more than large jobs under a small $\beta$ as they have higher priorities to be allocated resources in the cluster. In addition, when $\beta = .6$, the average job flowtime attains its minimum. 

As illustrated in Fig.~\ref{light_traffic_laps}, when $\lambda = 1$, almost all of the jobs in the cluster can complete within 200 units of time under different settings for $\beta$. When the job arrival rate is low, the jobs with small workloads can get large fractions of  shared resource under a small value of $\beta$. In this case, the benefit of redundancy is marginal and tuning down the value of $\beta$ does not help much for small jobs. However, in terms of the average job flowtime, a smaller $\beta$ leads to a worse performance. The reason behind is that a large job usually has a very small chance to obtain shared resource under a small value of $\beta$ in LAPS+R($\beta$) since it takes a long time to complete, resulting in a large flowtime. Though the number of large jobs is small, the amount of total job flowtime contributed but those large ones is significant.

\begin{figure}
\centering
\includegraphics[width=0.46\textwidth]{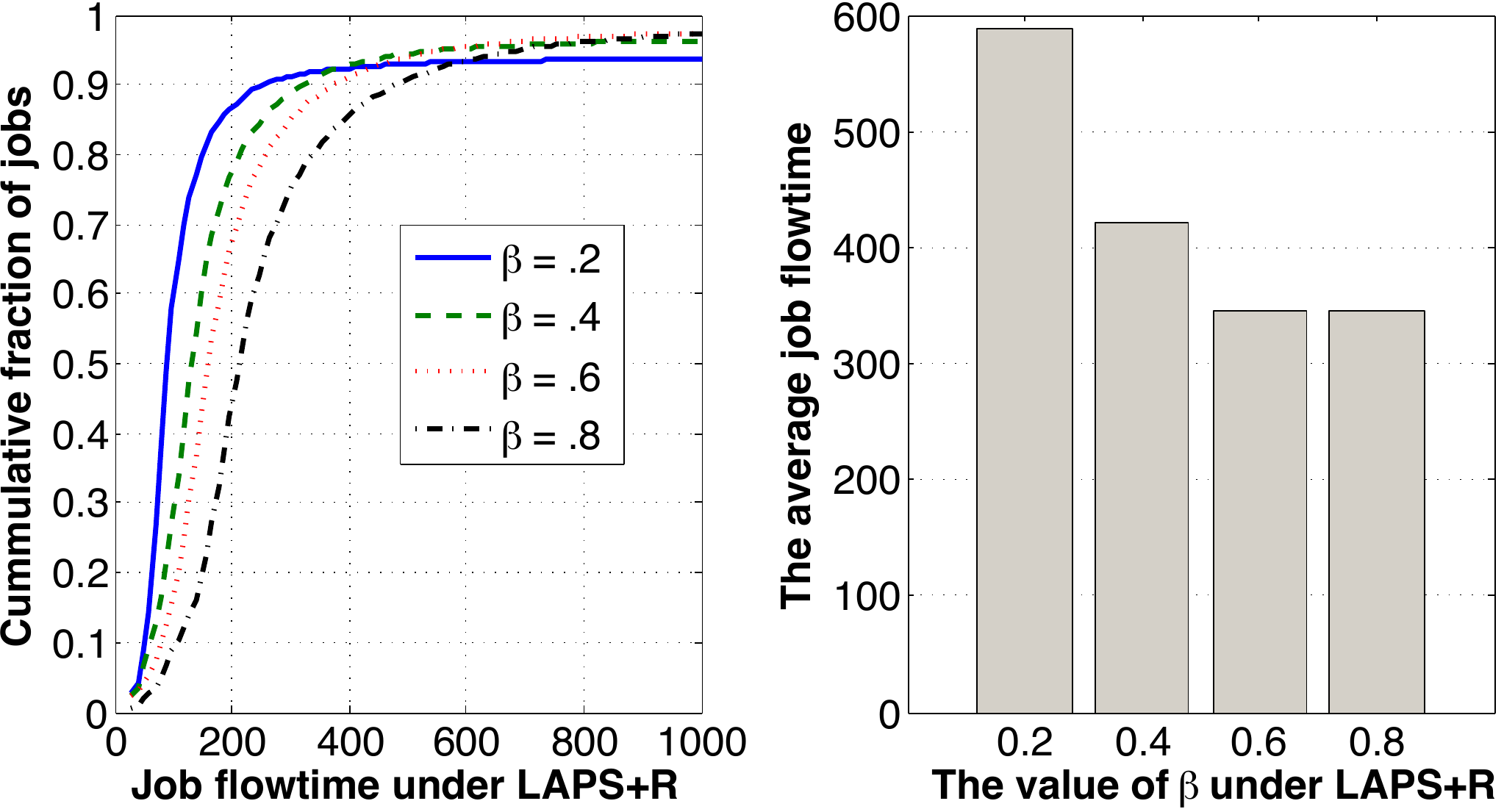}
\vspace{-0.8em}
\caption{The job flowtime under different $\beta$ in LAPS($\beta$) when $\lambda = 2.5$.}
\label{heavy_traffic_laps}
\vspace{ -1em}
\end{figure}

\begin{figure}
\centering
\includegraphics[width=0.46\textwidth]{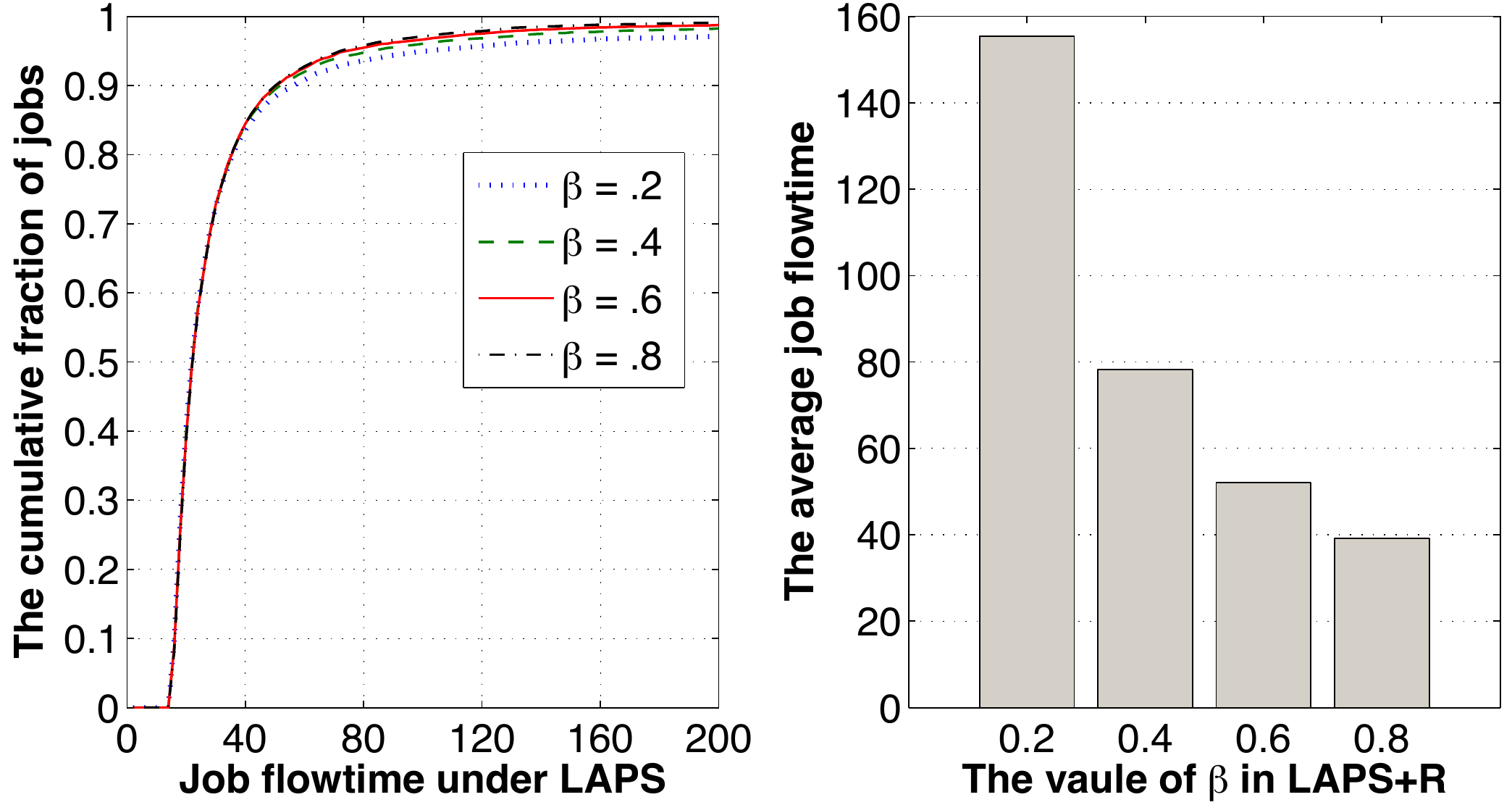}
\vspace{-0.1em}
\caption{The job flowtime under different $\beta$ in LAPS+R($\beta$) when $\lambda = 1$.}
\label{light_traffic_laps}
\vspace{-1.3 em}
\end{figure}

\section{Conclusions and Future Directions}
\label{conclusions}

This paper is an attempt to address the impact of two key sources of variability in parallel computing clusters: job processing times and machine processing rate. Our primary aim and contribution was to introduce a new speedup function account for redundancy, and provide the fundamental understanding on how job scheduling and redundant execution algorithms with limited number of checkpointings can help to mitigate the impact of variability on job response time. As the need of delivering predictable service on shared cluster and computing platforms grows, approaches, such as ours, will likely be an essential element of any possible solution. Extensions of this work to non-clairvoyant scenarios, the case of jobs with associated task graphs etc, are likely next steps towards developing the foundational theory and associated algorithms to address this problem.

\bibliographystyle{abbrv}
\bibliography{scheduling}  

\appendices

\section{Proof of Lemma \ref{lemma_approximation}}
\begin{proof}
Consider an optimal solution to OPT, ${\bm{y}^*}$, whose corresponding job completion time for job $j$ is denoted by $c^{*}_j$. Thus, for all $j=1,2,\cdots,N$, $\bm{y}^*$ and $c^{*}_j$ satisfy:
\begin{equation}
\int_{a_j}^{c^*_j}h_j(\bm{t_j}^*,\bm{r_j}^*,t)dt = p_j.
\end{equation} 
Moreover, it follows that $h_j(\bm{t_j}^*,\bm{r_j}^*,t) = 0$ for all $t \geq c_j^*$,  thus, we have that: 
\begin{equation}
 \int_{a_j}^{\infty}h_j(\bm{t_j}^*,\bm{r_j}^*,t)dt = p_j,
\end{equation}
and it follows that:
\begin{equation}
\begin{split}
& \int_{a_j}^{\infty}\frac{1}{p_j}({t - a_j}) h_j(\bm{t_j}^*,\bm{r_j}^*,t)dt = \int_{a_j}^{c_j^*}\frac{1}{p_j}({t - a_j}) h_j(\bm{t_j}^*,\bm{r_j}^*,t)dt  \\  & \leq \int_{a_j}^{c_j^*}\frac{1}{p_j}({c^*_j-a_j}) h_j(\bm{t_j}^*,\bm{r_j}^*,t)dt  = c^*_j - a_j.
 \end{split}
 \label{first_part}
\end{equation}
Following Lemma \ref{lemma_redundancy_progress}, it can be readily shown that $h_j(\bm{t_j}^*,\bm{r_j}^*,t) \leq \Delta$. Therefore, we have:
\begin{equation}
\label{second_part}
p_j = \int_{a_j}^{c^*_i}h_j(\bm{t_j}^*,\bm{r_j}^*,t)dt \leq \Delta (c^*_j - a_j).
\end{equation}
Combining \eqref{first_part} and \eqref{second_part}, we have:
$$\int_{a_j}^{\infty}\frac{({t-a_j + 2{p_j}})}{p_j} \cdot h_j(\bm{t_j}^*,\bm{r_j}^*,t)dt \leq (1 + 2\Delta) (c^*_j - a_j). $$
Since the optimal solution to OPT must be feasible for P1, it follows that:
\begin{equation}
\begin{split}
P1 & \leq \sum_{j=1}^{N}\int_{a_j}^{\infty}\frac{({t-a_j + 2{p_j}})}{p_j} \cdot h_j(\bm{t_j}^*,\bm{r_j}^*,t)dt \\ & \leq (1 + 2\Delta) \sum_{j=1}^{N}(c^*_j - a_j) = (1 + 2\Delta)OPT.
\end{split}
\end{equation}
This completes the proof. 
\end{proof}

\section{Proof of Lemma \ref{lemma_checking_fair}}
\begin{proof}
First, we have: 
\begin{equation}
\begin{split}
& \frac{1}{4(4+\epsilon)p_j}  \int_{a_j}^{c_j} \mathbbm{1}({n(t) < M}) \cdot \widetilde{h}_j(\bm{t}_j,\bm{x}_j,\bm{r}_{j},t) dt \\ & \leq \frac{1}{4(4+\epsilon)p_j}\int_{a_j}^{c_j}  \widetilde{h}_j(\bm{t}_j,\bm{x}_j,\bm{r}_{j},t)dt = \frac{1}{4}.
\end{split}
\end{equation}
Next, we proceed to show the following result holds. 
\begin{equation}
\label{equation_check_fea_proof}
\begin{split}
& \frac{\int_{a_j}^{c_j} \sum_{k:a_k \leq a_j} \mathbbm{1}(j \in A(\tau))  \cdot \mathbbm{1}(n(\tau) \geq M) \widetilde{h}_k(\bm{t}_{k},\bm{x}_{k},\bm{r}_{k},\tau) d\tau }{(4+\epsilon)Mp_j} \\  & \leq \frac{t-a_j}{p_j} +  \frac{1}{(4+\epsilon)M} n(t). 
\end{split}
\end{equation}
To achieve this, we divide the job set $\Psi_j = \{k:a_k \leq a_j\}$ into two separate sets: $\Psi^1_j = \{k: c_k \leq t \} \cap\Psi_j$ and $\Psi^2_j = \{k: c_k > t \} \cap\Psi_j$. For the first set, we have:
\begin{equation}
\label{check_feasbility_1}
\begin{split}
& \frac{\int_{a_j}^{c_j}\sum_{k: k \in \Psi^1_j} \mathbbm{1}(k \in A(\tau))  \cdot \mathbbm{1}(n(\tau) \geq M) \cdot \widetilde{h}_k(\bm{t}_{k},\bm{x}_{k},\bm{r}_{k},\tau) d\tau }{M}  \\  & \leq \frac{\int_{a_j}^{t}\sum_{k: k \in \Psi^1_j} \mathbbm{1}(k \in A(\tau))  \cdot \mathbbm{1}(n(\tau) \geq M) \cdot \widetilde{h}_k(\bm{t}_{k},\bm{x}_{k},\bm{r}_{k},\tau) d\tau}{M}.
\end{split}
\end{equation}

Based on the scheduling principle of Fair+R, it follows that:
\begin{equation}\sum_{k} \mathbbm{1}_{k \in A(t)}  \cdot \mathbbm{1}_{n(t) \geq M} \cdot \widetilde{h}_k(\bm{t}_k,\bm{x}_{k},\bm{r}_{k},t) \leq (4+\epsilon)M.
\vspace{-.3em}
\end{equation}
Therefore,  the L.H.S of \eqref{check_feasbility_1} is upper bounded by ${(4+\epsilon)(t-a_j)}$. For all jobs in $\Psi^2_j$, we have:
\begin{equation}
\label{check_feasbility_2}
\begin{split}
&  \int_{a_j}^{c_j}\sum_{k: k \in \Psi^2_j} \mathbbm{1}(k \in A(\tau))  \cdot \mathbbm{1}(n(\tau) \geq M) \cdot \widetilde{h}_k(\bm{t}_{k},\bm{x}_{k},\bm{r}_{k},\tau) d\tau \\  & =  \sum_{k: k \in \Psi^2_j} \int_{a_j}^{c_k} \mathbbm{1}(k \in A(\tau))  \cdot \mathbbm{1}(n(\tau) \geq M) \cdot \widetilde{h}_k(\bm{t}_{k},\bm{x}_{k},\bm{r}_{k},\tau)d\tau \\ & \overset{\mbox{(ii)}} \leq  \sum_{k: a_k \leq t \leq c_k \leq c_j} p_j \leq n(t)p_j,
\end{split}
\end{equation}
where (ii) is due to that, for any job who arrives before $j$, its amount of work processed between the range $[a_j,c_k]$ is upper bounded by $p_j$. 

Combining all inequalities above, the lemma immediately follows. This completes the proof. 
\end{proof}

\section{Proof of Lemma \ref{lemma_dual_performance} }
\begin{proof}
First, it can be readily shown that:
\begin{equation}
\label{bound_beta}
M\int_{0}^{\infty}\beta(t)dt = \frac{1}{4+\epsilon}n(t)dt = \frac{RF}{4+\epsilon}.
\end{equation}

Next, we proceed to show $\sum_{j=1}^{N}\alpha_j \geq \frac{RF}{4}$. To achieve this, we consider the following two cases:

\vspace{.5em}
\textbf{Case I: $n(t) \geq M$}. In this case, it's easy to verify that $\alpha_j(t) = 0$ for $j \leq l$ and $\alpha_j(t) = \frac{j-l}{kMp_j}$ for $l < j < n(t)$. Therefore, it follows that:
\begin{equation}
\sum_{j=1}^{N}\alpha_j(t) p_j  = \sum_{j = l+1}^{n(t)}\frac{j-l}{kM} = \frac{kM+1}{2} \geq \frac{n(t)}{4}
\end{equation}

\textbf{Case II: $n(t) < M$}. In this case, we have: $\widetilde{h}_j(\bm{t}_j,\bm{x}_j,\bm{r}_{j},t) \geq 4+\epsilon$ since we are using a resource augmentation of $(4+\epsilon)$-speed. Hence, the following equation holds:
\begin{equation}
\sum_{j=1}^{N}\alpha_j(t) p_j \geq \frac{1}{4}\sum_{j=1}^{n(t)} \mathbbm{1}(n(t) < M) = \frac{n(t)}{4}.
\end{equation}
\vspace{.5em}
As such, we have:
\begin{equation}
\label{bound_alpha}
\begin{split}
\sum_{j=1}^{N}\alpha_j p_j & = \sum_{j=1}^{N}\int_{a_j}^{c_j} \alpha_{j}(\tau) p_j d\tau = \int_{0}^{\infty} \sum_{j=1}^{N} \alpha_{j}(\tau) p_j d\tau \\  & \geq \frac{1}{4} \int_{0}^{\infty}  n(\tau) d\tau =  \frac{RF}{4}.
 \end{split}
\end{equation}
The result follows combining \eqref{bound_beta} and \eqref{bound_alpha}. This completes the proof. 
\end{proof}

\ifCLASSOPTIONcaptionsoff
  \newpage
\fi

\end{document}